%% file: MS28075manuscript.tex
\documentclass[12pt,oneside,reqno]{amsart}
\usepackage[foot]{amsaddr}
\usepackage[latin9]{inputenc}
\usepackage{color}
\definecolor{note_fontcolor}{rgb}{1, 0, 0}
\usepackage{textcomp}
\usepackage{mathrsfs}
\usepackage{bm}
\usepackage{pdfpages}
\usepackage{amsbsy}
\usepackage{amstext}
\usepackage{amsthm}
\usepackage{amssymb}
\usepackage{graphicx}
\usepackage{setspace}
\usepackage[authoryear]{natbib}
\usepackage[unicode=true,
 bookmarks=false,
 breaklinks=false,pdfborder={0 0 1},backref=false,colorlinks=false]
 {hyperref}
\usepackage{caption} 
\makeatletter

\numberwithin{equation}{section}
\numberwithin{figure}{section}
\theoremstyle{definition}
 \newtheorem{example}{\protect\examplename}
\theoremstyle{definition}
\newtheorem{prop}{\protect\propositionname}
\theoremstyle{definition}
\newtheorem{fact}{\protect\factname}
\theoremstyle{definition}
\newtheorem{defn}{\protect\definitionname}
\theoremstyle{definition} \newtheorem{thm}{\protect\theoremname}
\theoremstyle{definition}
\newtheorem{cor}{\protect\corollaryname}
\theoremstyle{definition}
\newtheorem{lem}{\protect\lemmaname}
\theoremstyle{definition}
\newtheorem{rem}{\protect\remarkname}

\usepackage{amsfonts}
\usepackage{bm}
\usepackage{setspace}

\usepackage{graphics}
\usepackage{epsfig}\usepackage{verbatim}\usepackage{bm}\usepackage{latexsym}\usepackage{url}\usepackage{rotating}

\usepackage{mathrsfs}
\usepackage{multirow}
\usepackage{array}
\usepackage{url}

\usepackage{graphicx}  
\makeatletter 
\def\paragraph{\@startsection{paragraph}{4}%
 \z@\z@{-\fontdimen2\font}%
{\normalfont\bfseries}} 
\makeatother

\newcommand{\SInf}{\operatorname{SI}}
\newcommand{\TI}{\operatorname{TI}}
\newcommand{\Covar}{\operatorname{Cov}}
\newcommand{\Vtext}{\operatorname{V}}

\usepackage[usenames,dvipsnames,svgnames,table]{xcolor}
\usepackage[margin=.97in]{geometry}
\hypersetup{colorlinks=true, urlcolor=MidnightBlue, citecolor=MidnightBlue, linkcolor=MidnightBlue, linkbordercolor = white}

\usepackage{tikz}
\usetikzlibrary{arrows.meta}

\usepackage{xr}
\makeatletter
\newcommand*{\addFileDependency}[1]{
  \typeout{(#1)}
  \@addtofilelist{#1}
  \IfFileExists{#1}{}{\typeout{No file #1.}}
}
\makeatother

\newcommand*{\myexternaldocument}[2]{
    \externaldocument[#2]{#1}
    \addFileDependency{#1.tex}
    \addFileDependency{#1.aux}
}

\myexternaldocument{appendix}{OA-}
\makeatother

\providecommand{\corollaryname}{Corollary}
\providecommand{\definitionname}{Definition}
\providecommand{\examplename}{Example}
\providecommand{\factname}{Fact}
\providecommand{\lemmaname}{Lemma}
\providecommand{\propositionname}{Proposition}
\providecommand{\theoremname}{Theorem}
\providecommand{\remarkname}{Remark}

\usepackage{siunitx,booktabs}

\title{Learning from Neighbors about a Changing State}
\thanks{
We gratefully acknowledge funding from Pershing Square Fund for Research on the Foundations of Human Behavior at Harvard, the Rockefeller Foundation (Dasaratha), and the National Science Foundation under grants SES-1847860 and SES-1629446 (Golub), as well as the hospitality of the economics department at the University of Pennsylvania during a key phase of this work. Hershdeep Chopra, Yu-Chi Hsieh, Yixi
	Jiang, Joey Feffer, and Rithvik Rao provided excellent research assistance.
	For valuable comments we are grateful to (in random order) Erik Madsen, Alireza Tahbaz-Salehi,	Drew Fudenberg, Matt Elliott, Nageeb Ali, Tomasz Strzalecki, Alex
	Frankel, Yair Livne, Jeroen Swinkels, Margaret Meyer, Philipp Strack, Erik Eyster,
	Michael Powell, Michael Ostrovsky, Leeat Yariv, Andrea Galeotti, Eric
	Maskin, Elliot Lipnowski, Jeff Ely, Eddie Dekel, Annie Liang, Iosif
	Pinelis, Ozan Candogan, Ariel Rubinstein, Paul Goldsmith-Pinkham, Bob Wilson, Omer Tamuz,
	Jeff Zwiebel, Matthew O. Jackson, Matthew Rabin, and Andrzej Skrzypacz. We also thank Evan Sadler for important
	conversations early in the project, David Hirshleifer for detailed
	comments on a draft, and Rohan Dutta and Kevin He for helpful discussions. The paper was greatly improved by advice from four anonymous referees and the editor, Adam Szeidl. All errors are our own.}
\author{Krishna Dasaratha}
\address[KD]{Boston University. Email: \texttt{\protect\href{mailto:krishnadasaratha\%40gmail.com}{krishnadasaratha@gmail.com}}}

\author{Benjamin Golub}
\address[BG]{Northwestern University. Email: \texttt{\protect\href{mailto:benjamin.golub\%40northwestern.edu}{benjamin.golub@northwestern.edu}}}

\author{Nir Hak}
\address[NH]{Uber Technologies, Inc. Email: \texttt{\protect\href{mailto:nirhak\%40gmail.com}{nirhak@gmail.com}}}

\date{\today}

\begin{document}

\begin{abstract}
Agents learn about a changing state using private signals and their neighbors' past estimates of the state. We present a model in which Bayesian agents in equilibrium use neighbors' estimates simply by taking weighted sums with time-invariant weights. The dynamics thus parallel those of the tractable DeGroot model of learning in networks, but arise as an equilibrium outcome rather than a behavioral assumption. We examine whether information aggregation is nearly optimal as neighborhoods grow large. A key condition for this is \emph{signal diversity}: each individual's neighbors have private signals that not only contain independent information, but also have sufficiently different distributions.  Without signal diversity---e.g., if private signals are i.i.d.---learning is suboptimal in all networks and highly inefficient in some.  Turning to social influence, we find it is much more sensitive to one's signal quality than to one's number of neighbors, in contrast to standard models with exogenous updating rules.

\end{abstract}
\maketitle

\thispagestyle{empty}
\setcounter{page}{0} \clearpage

\section{Introduction}

\setcounter{page}{1}

People learn from others about conditions relevant to decisions they
have to make. In many
cases, the conditions---e.g., the state of a labor market facing workers, or the business environment relevant to an organization---are changing. Thus, welfare depends not on learning a static ``state of the world,'' but rather on staying up to date with a changing state. The phenomenon of adaptation and responsiveness to new information is central in many economic applications, including in economic development, the study of organizations, and financial decision-making. When is a group of agents successful, collectively, in adapting efficiently to a changing environment? The answers lie partly in the structure of the social networks that shape
agents' social learning opportunities. Our model is designed to analyze
how a group's adaptability is shaped by the properties of such networks,
the inflows of information into society, and the interplay of the
two.

We consider overlapping generations of agents  who are interested in
tracking an unobserved \emph{state} that evolves over time.\footnote{Cf. \citet{banerjee2004word}
and \citet{wolitzky2018learning}, with overlapping generations and a fixed state.} The state
is a Gaussian AR(1) process: somewhat persistent, but with constant innovations to learn about. Each
agent, before making a decision, engages in social learning: she observes
the actions of some members of prior generations, which reveal their
estimates of the state. The social learning opportunities are embedded in a network, in that
one's network position determines the \emph{neighborhood} of peers
whom one observes. Neighborhoods reflect geographic, cultural, organizational,
or other kinds of proximity.\footnote{\citet{sethi2019culture} argue that, even without explicit
communication costs or constraints, familiarity and shared context determine the network in which people can effectively communicate.} In addition to social information, agents
also receive private signals about the current state, with Gaussian distributions
that may also vary with network position; in particular, some agents
may receive more precise information about the state than others.

We  give some examples. When a university student begins searching
for jobs, she becomes interested in the state of the relevant
labor market (e.g., typical wages for someone like her), which naturally
vary over time. She uses her own private research (a private signal) but also learns
from the choices of others (e.g., recent graduates) who have recently
faced a similar problem. Whom she can learn from will depend on her academic specialization, dormitory, extracurricular activities, and so forth: she will predominantly observe
predecessors who are ``nearby'' in these ways.  Similarly, when a new cohort of professionals enters a firm such as
a management consultancy or law practice, they learn about the business
environment from their seniors. Who works with whom, and therefore
who learns from whom, is shaped by the structure of the organization.  Beyond heterogeneity in network position, agents differ in the precision of the private signals they can access from outside the network: for example, people with quantitative training may be better placed to learn from external statistical reports.

We now detail our three main contributions. The \emph{first} is to develop a tractable model of learning about a changing state in which Bayesian updating has a simple form: each agent forms an estimate by taking a weighted average of neighbors' estimates and her private signal. Because the environment is stationary, the weights are stationary as well.  Equilibrium behavior thus yields a dynamic paralleling the \citet{degroot1974reaching} network learning model, which is famous for its tractability, but which has been criticized for its lack of canonical foundations \citep{molavi2018theory,golub-sadler}.  The weights in agents' learning rules are endogenously
determined because, when each agent extracts information from neighbors'
estimates, the information content of those estimates depends on the
neighbors' learning rules. We characterize these weights and the distributions of behavior in a stationary equilibrium.\footnote{In Bayesian models of learning about a fixed state there is a time-dependence whereby rational updating rules depend on the time elapsed since the learning process started. Studies of such models often focus on an eventual rate of learning about a fixed state. See, for instance, \citet*{molavi2018theory} and \citet*{harel2021rational}. This time-dependence is absent in our stationary environment; equilibrium outcomes can be summarized by steady-state updating weights and error rates.} These characterizations permit the analysis of comparative statics of behavior and welfare. The model is well-behaved computationally: equilibria can be calculated quickly in networks of thousands of nodes,
which makes the framework useful for structural exercises.  Finally, the basic framework and stationary equilibrium definition extend readily to accommodate various kinds of behavioral updating rules, e.g., ones that neglect correlations in neighbors' actions.

Our \emph{second} contribution is to analyze when equilibrium learning facilitates good information aggregation. Here we have positive and negative results. The main positive finding is that Bayesian agents in equilibrium can achieve {good aggregation} under a \emph{signal diversity} condition. To formalize ``good aggregation,'' we first note that in our model social information is valuable to agents insofar as it allows them to estimate the  state before the current period; our measure of aggregation quality is the accuracy of these estimates. We say {good aggregation} occurs if each agent has an estimate nearly as precise as if she simply knew the previous state. We say \emph{signal diversity} holds if each individual has at least two different private signal precisions represented in large numbers in her neighborhood. The positive result says that signal diversity is sufficient for good aggregation, robustly across a large class of random networks (ones arising from stochastic block models satisfying certain technical conditions).  Signal diversity is valuable because it leads to
diversity of neighbors' strategies: it makes them use recent and older information differently from one another. An agent observing such neighbors, in turn, can use the diversity for statistical identification, constructing a precise  estimate of the most recent state. We illustrate this key idea in an example at the end of the introduction.

We complement the positive finding with two negative results showing that both the ``signal diversity'' and the ``Bayesian'' conditions are important. First, signal {homogeneity} turns out to be a fundamental obstruction to good aggregation. In an environment where everyone's private signals are conditionally independent and identically distributed, equilibrium aggregation is bound to be inefficient. We begin by establishing this
point in highly symmetric networks, where the failure of aggregation is shown to have severe welfare consequences, making each agent worse off by an unbounded amount relative to a world with signal diversity. A more general finding is that  in large networks with signal homogeneity, it is impossible in equilibrium to achieve accuracies
of aggregation of the same order as in our positive
result under signal diversity. One might have thought that diversity 
of neighbors' network positions (and thus learning opportunities) can substitute for diversity of their private signal distributions; our negative result shows that network diversity is a poor substitute for signal diversity.

We next show that the ``Bayesian'' condition is also important for the good aggregation result. To do this, we contrast the learning of  Bayesians best-responding to others' learning rules with that of a naive population too unsophisticated to account for correlations in neighbors' learning errors, as in some canonical behavioral learning models \citep{eyster2010naive}.\footnote{See also \citet{bala1998learning}, a seminal model of boundedly rational
learning rules in networks.} We identify a class of such models in which  information aggregation is essentially guaranteed to fall
short of good aggregation benchmarks for all agents. The deficiencies
of naive learning rules are different from and more severe than those
in similar problems with an unchanging state, where naive heuristics can aggregate information very well.\footnote{In analogous fixed-state environments
where individuals have sufficiently many observations, if everyone
uses certain simple and stationary DeGroot-style heuristics (requiring no sophistication
about correlations between neighbors' behavior), they can learn the
state quite precisely \citep*{Golub2010Naive,Jadbabaie2013NonBayesian}.
A changing state makes such imitative heuristics quite inefficient.}

Our \emph{third} contribution is to study social influence, an outcome of central importance in network theory. We define a notion of steady-state social influence\textemdash how
an idiosyncratic change in an individual's information affects others' average
behavior. This is analogous to the definition of influence in the standard
DeGroot model (where updating weights are exogenous). 
The endogenous determination of
weights makes a big difference for how the environment affects social
influence.
Relative to the DeGroot model benchmark studied in \cite*{DeMarzo2003Persuasion}, an agent's social influence is much
more sensitive to the quality of her private information. On the other hand, just as in the standard benchmark, an agent's influence is approximately proportional to her degree. 

A closing discussion makes two main points. First, some of our theoretical aggregation results use large random graphs. We perform a numerical
exercise to show that the main message about information aggregation\textemdash diversity of signal types helps learning\textemdash remains valid when we calculate equilibria on graphs reflecting real social networks.  Second,  our analysis generalizes readily to richer models of multidimensional states and signals. As one application of such a generalization, we  consider a manager 
who wishes to facilitate better learning in an organization, and ask what distributions of
expertise such a designer would prefer. Our results provide a
distinctive rationale for \emph{informational specialization} as a design feature that facilitates good information aggregation. 

\medskip{}

\paragraph*{An example}

We now present a simple example that illustrates our dynamic model,
highlights obstacles to learning that distinctively arise in a dynamic
environment, and gives a sense of some of the main forces that play
a role in our results on the quality of learning.

Consider a particular environment, with a single perfectly informed
source $S$; many media outlets $M_{1},\ldots,M_{n}$ with access
to the source as well as some independent private information; and
the general public. The public consists of many individuals who learn
only from the media outlets. We are interested in how well each member
of the public could learn by following many media outlets. More precisely,
we consider the example shown in Figure \ref{Fig:introexample} and
think of $P$ as a generic member of the large public.

\begin{figure}[th]
\centering{}\includegraphics[width=3in]{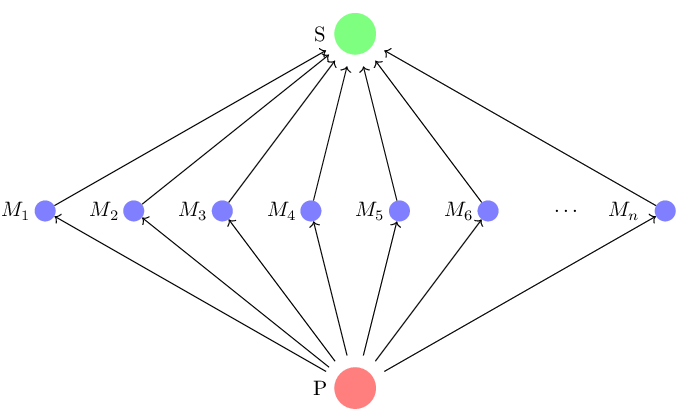}
\begin{centering}
\caption{\label{Fig:introexample}\footnotesize The network used in the ``value of diversity'' example}
\par\end{centering}
\end{figure}

The state $\theta_{t}$ follows a Gaussian random walk: $\theta_{t}=\theta_{t-1}+\nu_{t}$,
where the innovations $\nu_{t}$ are standard normal. Each period,
the source learns the state $\theta_{t}$ and takes an action (which can be thought of simply as making an announcement) that reveals it. The media outlets observe the source's action
from the previous period, which is $\theta_{t-1}$. At each time period,
they also receive noisy private signals, $s_{M_{i},t}=\theta_{t}+\eta_{M_{i},t}$
with normally distributed, independent, mean-zero errors $\eta_{M_{i},t}$.
They then announce their posterior means of $\theta_{t}$, which we
denote by $a_{M_{i},t}$. The member of the public, in a given period
$t$, makes an estimate based on the observations $a_{M_{1},t-1},\ldots,a_{M_{n},t-1}$
of media outlets' actions in the previous period. All agents are short-lived:
they see actions in their neighborhoods one period ago, and then they
take an action that reveals their posterior belief of the state.

If we had an \emph{unchanging} state  but the same signals and observation
structure, learning would trivially be perfect: media outlets would
learn the state from the source and report it to the public. In the dynamic environment, given that $P$ has no signal,
she can at best hope to learn $\theta_{t-1}$
(and use that to estimate $\theta_{t}$). Can this benchmark
be achieved, and if so, when?

A typical estimate of a media outlet at time $t$ is a linear combination
of $s_{M_{i},t}$ and $\theta_{t-1}$ (the latter being the
information that the media outlets learned from the source $S$). In particular,
the estimate of media outlet $M_i$ at time $t$ can be expressed as 
\[
a_{M_{i},t}=w_{i}s_{M_{i},t}+(1-w_{i})\theta_{t-1},
\]
where the weight $w_{i}$ on the media outlet's signal is increasing
in the precision of that signal. We give the public no private signal,
for simplicity only.

Suppose first that the media outlets have identically distributed
private signals. Because the member of the public observes many symmetric
media outlets, it turns out that her best estimate of the state, $a_{P,t}$,
is simply the average of the estimates of the media outlets. Since
each of these outlets uses the same weight $w_{i}=w$ on its private
signal, we may write

\[
a_{P,t}=w\sum_{i=1}^{n}\frac{s_{M_{i},t-1}}{n}+(1-w)\theta_{t-2}\approx w\theta_{t-1}+(1-w)\theta_{t-2}.
\]
That is, $P$'s estimate is an average of media private signals from
last period ($t-1$), combined with what the media learned from the source,
which tracks the state in the period before that ($t-2$). In the approximate
equality, we have used the fact that an average of many private signals
is approximately equal to the state, by our assumption of independent
errors. No matter how many media outlets there are, and even though
each has independent information about $\theta_{t-1}$, the public's
beliefs are confounded by older information.

What if, instead, half of the media outlets (say $M_{1},\ldots,M_{n/2}$)
have more precise private signals than the other half, perhaps because
these outlets have invested more in expertise on this topic? The
media outlets with more precise signals, called group $A$, will then place weight $w_{A}$
on their private signals, while the media outlets with less precise
signals (group $B$) use a smaller weight, $w_{B}$. We will now argue that a member
of the public can extract more information from the media in this
setting. In particular, she can first compute the averages of the
two groups' actions
\setlength{\fboxsep}{0pt}
\setlength{\fboxrule}{0pt}
\begin{align*}
\text{\fbox{%
    \parbox{1.2in}{\footnotesize \raggedright type $A$ average action at time $t-1$}}} &=w_{A}\; \;\, \sum_{i=1}^{n/2}\;\;\,\frac{s_{M_{i},t-1}}{n/2}+(1-w_{A})\theta_{t-2} \; \approx \; w_{A}\theta_{t-1}+(1-w_{A})\theta_{t-2} \\
\text{\fbox{%
    \parbox{1.2in}{\footnotesize \raggedright type $B$ average action at time $t-1$}}} &=w_{B}\sum_{i=n/2+1}^{n}\frac{s_{M_{i},t-1}}{n/2}+(1-w_{B})\theta_{t-2} \; \approx \; w_{B}\theta_{t-1}+(1-w_{B})\theta_{t-2}.
\end{align*}
Then, since $w_{A}>w_{B}$, the public knows two distinct linear combinations
of $\theta_{t-1}$ and $\theta_{t-2}$. The state $\theta_{t-1}$
is identified from these. So the member of the public can form a very
precise estimate of $\theta_{t-1}$\textemdash which, recall, is as
well as she can hope to do. The key force is that the two groups of
media outlets give different mixes of the old information and the
more recent state, and by understanding this, the public can infer
both.  Indeed, to recover $\theta_{t-1}$, the public puts a negative weight on the actions of media outlets of type $B$, which allows it to subtract off old information and focus on the recent state, $\theta_{t-1}$.  One can show that if, in contrast, agents are naive, e.g., if they think that all
of the estimates of the media are uncorrelated (or only mildly correlated) conditional
on the state, they will put positive weights on their observations
and will again be bounded from learning the state.

This illustration relied on a particular network with several special
features: a very ``central'' source, one-directional links, and no communication among
the media outlets or the public. We will show that the same considerations
determine learning quality in a large class of random networks in
which agents have many neighbors, with complex connections among them.
Quite generally, if each neighborhood contains a diversity of signal types, agents can concentrate on new developments in the state while
filtering out old, less relevant information and thus estimate the
changing state as accurately as physical constraints allow.

\subsection*{Outline}

Section \ref{sec:Model} sets up the basic model and discusses its
interpretation. Section \ref{sec:Equilibrium} defines equilibrium, shows its existence, and characterizes it. Section \ref{sec:LearningOutcomes}
reports our main theoretical results on the quality of
information aggregation. In Section \ref{sec:Importance-Anti-Im},
we discuss learning outcomes under a variety of non-Bayesian models.
Section \ref{sec:socialinfluence} defines and analyzes social influence. Section \ref{sec:Related-literature}
relates our model and results to the social learning literature. In
Section \ref{sec:Discussion-and-extensions}, we describe our numerical exercise with network data from Indian villages and discuss a simple extension to multi-dimensional states to interpret our results on signal diversity.

\section{Model\label{sec:Model}}

\smallskip{}

\paragraph{State of the world}

There is a discrete set of instants, $\mathcal{T}=\mathbb{Z}=\left\{ \ldots,-2,-1,0,1,2,\ldots\right\} .$ At each time $t\in\mathcal{T}$, there is a \emph{state}, a random
variable $\theta_{t}$ taking values in $\mathbb{R}$. This state
evolves as an AR(1) stochastic process. That is, 
\begin{equation}
\theta_{t+1}=\rho\theta_{t}+\nu_{t+1},\label{eq:ar1}
\end{equation}
where $\rho$ is a constant with $0<|\rho|\leq1$ and $\nu_{t+1}\sim\mathcal{N}(0,\sigma_{\nu}^{2})$
are independent \emph{innovations}. 
When $|\rho|<1$ we have the explicit formula
\[
\theta_{t}=\sum_{\ell=0}^{\infty}\rho^{\ell}\nu_{t-\ell},
\]
and thus  the state at any time $t$ has the stationary distribution $\theta_{t}\sim\mathcal{N}\left(0,\frac{\sigma_{\nu}^{2}}{1-\rho^{2}}\right).$
We maintain the normalization throughout  that innovations have variance $1$, i.e., $\sigma_{\nu}=1$.

We will occasionally examine an alternative specification (making our departure from the main model explicit) where there is a starting time, so that $\mathcal{T}=\mathbb{Z}_{\geq0}=\left\{ 0,1,2,\ldots\right\} $, and the state process is defined as in (\ref{eq:ar1}) starting at time $0$ with some specified distribution for $\theta_0$.

\smallskip{}

\paragraph{Information and observations}

There is a set of \emph{nodes} is $N=\left\{ 1,2,\ldots,n\right\} $. Each
node $i$ can be thought of as a location, and is associated with a set $N_{i}\subseteq N$ of nodes that
$i$ can observe, called its \emph{neighborhood}.\footnote{For all results, a node $i$'s neighborhood can, but need not, include
$i$ itself.}

Each node is populated by a sequence of \emph{agents} in overlapping
generations. For each time $t$, there is a node-$i$ agent, labeled
$(i,t)$, who takes that node's time-$t$ action $a_{i,t}$. This agent is born
at time $t-m$ and has $m$ periods to
observe the actions taken in her neighborhood before she acts. Thus, when taking
her action, the agent $(i,t)$ knows $a_{j,t-\ell}$ for all nodes
$j\in N_{i}$ and all lags $\ell\in\left\{ 1,2,\ldots,m\right\} $. We
call $m$ the \emph{memory}; it reflects how many periods of 
actions in her neighborhood  an agent passively observes before acting. (See Figure  \ref{fig:timeline} for an illustration.) One interpretation
is that a node corresponds to a role in an organization. A worker
in that role has some time to observe colleagues in related roles
before choosing a once-and-for-all action herself. Much of our analysis
is done for an arbitrary finite $m$; we view the restriction to finite
memory as useful for avoiding technical complications, but because
$m$ can be arbitrarily large, this restriction has little substantive
content.\footnote{It is worth noting that even when the memory $m$ is small, observed actions can indirectly incorporate signals from much further in the past.}

\begin{figure} 
\begin{centering}
\resizebox{6in}{!}{%
\input{supportingfiles/timeline}
}
\par\end{centering}
\begin{centering}
\caption{\footnotesize An illustration, with memory $m=1$, of the overlapping generations at a node $i$. At time $t-1$, agent $(i,t)$ is born and observes contemporaneous actions in her neighborhood.  At time $t$, she observes her private signal $s_{i,t}$ and takes her action $a_{i,t}$. \label{fig:timeline}}
\par\end{centering}
\end{figure}

In addition to social information from her neighborhood, each agent also observes a private signal, 
\[
s_{i,t}=\theta_{t}+\eta_{i,t},
\]
where the error term $\eta_{i,t}\sim\mathcal{N}(0,\sigma_{i}^{2})$ has a variance
$\sigma_{i}^{2}>0$ that depends on the node but not on the time period.
All the errors $\eta_{i,t}$ and state innovations $\nu_{t}$ are independent of one another.
An agent's information is a vector consisting of her private signal
and all of her social observations. An important special case will
be $m=1$, where agents observe only one period of others' actions
before acting themselves, so that the agent's information is\emph{
$(s_{i,t},(a_{j,t-1})_{j\in N_{i}})$.}

The \emph{network} $G=(N,E)$ is
the set of nodes $N$ together with the (fixed) set of \emph{links
$E$}, defined as the subset of pairs $(i,j)\in N\times N$ such that
$j\in N_{i}$. This network (also called a \emph{graph}), which determines the observation structure, is common knowledge, as is the informational
environment (i.e., the joint distribution of all exogenous random variables). 

An \emph{environment} is specified by $(G,\bm{\sigma})$, where $\bm{\sigma}=(\sigma_{i})_{i\in N}$
is the profile of signal variances.

\smallskip{}

\paragraph{Preferences and best responses}

When an agent $(i,t)$ makes her
once-and-for-all choice $a_{i,t}\in\mathbb{R}$, her utility is given by 
\begin{equation} \label{eq:utility}
u_{i,t}(a_{i,t})=-\mathbb{E}[(a_{i,t}-\theta_{t})^{2}].
\end{equation}
By a standard fact about squared-error loss functions, given the distribution
of $(\bm{a}_{N_{i},t-\ell})_{\ell=1}^{m}$, the optimal choice of agent $(i,t)$ is to set her action equal to her expectation of the state:
\begin{equation}
a_{i,t}=\mathbb{E}[\theta_{t}\mid \underbrace{(\bm{a}_{N_{i},t-\ell})_{\ell=1}^{m},s_{i,t}}_{\text{\tiny $i$'s information}}].\label{eq:BR}
\end{equation}
Here the notation $\bm{a}_{N_{i},t'}$ refers to the vector $(a_{j,t'})_{j\in N_{i}}$ of time-$t'$ actions in the agent's neighborhood.
An action can be interpreted as an agent's estimate of the state,
and we will sometimes use this terminology.

The conditional expectation (\ref{eq:BR}) depends, of course, on
the prior of agent $(i,t)$ about $\theta_{t}$, which, under correctly
specified beliefs, has distribution $\theta_{t}\sim\mathcal{N}\left(0,\frac{\sigma_{\nu}^{2}}{1-\rho^{2}}\right)$.
We allow the prior to be any normal distribution or a diffuse improper prior.\footnote{We take priors about $\theta_t$, like the information structure and network, to be
common knowledge.} It saves on notation to analyze the case where all agents have improper
priors. Because actions under a normal prior are related to actions
under the improper prior by a simple linear bijection\textemdash and
thus have the same information content for other agents\textemdash all
results immediately extend to the general case.

The doubly-infinite time axis introduces some subtleties into the definition of strategy profiles; complete details are formalized in online Appendix \ref{OA-sec:Details-of-definitions}.

\section{Updating and equilibrium\label{sec:Equilibrium}}

In this section we study agents' learning behavior and present a notion
of stationary equilibrium. We begin with the canonical case of Bayesian
agents with correct models of others' behavior; we study other behavioral
assumptions in Section \ref{sec:Importance-Anti-Im} below.

\subsection{Best-response behavior} \label{subsec:Best-Response-Weights}

The first step is to analyze optimal updating behavior in response to 
others' strategies. A strategy of an agent is \emph{linear} if the action taken is a linear
function of the variables she observes. We will analyze
agents' best responses to linear strategies, showing that they are
linear and computing them explicitly.\footnote{This analysis applies both to the main $\mathcal{T}=\mathbb{Z}$ model and the alternative with $\mathcal{T}=\mathbb{Z}_{\geq0}$. For a discussion of why it is natural to consider linear opponent strategies, see Section \ref{sec:Stationary_eqm_linear_strategies} below.}

Fix an agent $(i,t)$ and some
linear strategy profile played before time $t$.  By (\ref{eq:BR}), this agent's best-response action $a_{i,t}$ is her conditional expectation of $\theta_t$ given her information, $\mathbb{E}[\theta_{t}\mid (\bm{a}_{N_{i},t-\ell})_{\ell=1}^{m},s_{i,t}]$. Each action before time $t$ can be written as a (possibly infinite) sum of past signals $s_{j,t'}$. It follows that all random variables appearing in the conditional expectation are jointly Gaussian.  That implies that $a_{i,t}=\mathbb{E}[\theta_{t}\mid (\bm{a}_{N_{i},t-\ell})_{\ell=1}^{m},s_{i,t}]$
is an affine function of $s_{i,t}$ and $( \bm{a}_{N_{i},t-\ell})_{\ell=1}^{m}$ (see \citealt*[Section 4.3]{eaton1983multivariate}).
We now analyze this conditional expectation in detail.

\subsubsection{Actions as estimates of states: A key covariance matrix} 
Agents learn partly from past agents' actions, so the joint distribution of actions as estimates of the state at an arbitrary time $t$ will be important to track.  An arbitrary agent's scaled past action $\rho^\ell a_{i,t-\ell}$ gives an estimate of the state in the sense that $\mathbb{E}[\theta_{t} \mid \rho^\ell a_{i,t-\ell}]=\rho^\ell a_{i,t-\ell}$. Turning to second moments, define the covariance matrix of the errors in these estimates over the most recent $m$ periods: $$\bm{V}_{t} = \Covar\left(\left(\rho^{\ell}a_{i,t-\ell}-\theta_{t} \right)_{\substack{i\in N \\ 0\leq\ell\leq m-1}} \right).$$ In the case $m=1$, we have $\bm{V}_{t}=\Covar\left(\left( a_{i,t}-\theta_{t}\right)_{i\in N}\right)$. 
We will often refer to $\bm{V}_{t}$ simply as the covariance matrix of the model, as it will play a central role in our subsequent analysis.

\subsubsection{Best-response weights} \label{sec:BRweights}

The information of agent $(i,t)$ at time $t$ may be represented as a random vector $$\bm{z}_{i,t}=\left(\left(\rho^{\ell}a_{j,t-\ell}\right)_{\substack{j\in N_i \\ 1\leq\ell\leq m}}, s_{i,t} \right).$$ We will calculate the conditional expectation $\mathbb{E}[\theta_{t} \mid \bm{z}_{i,t}]$ in terms of a covariance matrix constructed from $(i,t)$'s observations,
\[
\bm{C}_{i,t-1}  = \Covar(\bm{z}_{i,t}-{\theta}_{t}\bm{1}) =\left(\begin{array}{cc}
\rho^2 \bm{V}_{N_{i},t-1} + \bm{1} \bm{1}^{\top} &  \bm{0}\\
 \bm{0}  & \sigma_{i}^{2}
\end{array}\right)\!,
\] where $\bm{V}_{N_{i},t-1}$ is the submatrix of $\bm{V}_{t-1}$ corresponding to indices in  $i$'s neighborhood.\footnote{We rewrite $\rho^{\ell}a_{j,t-\ell}-\theta_t=\rho(\rho^{\ell-1}a_{j,t-\ell}-\theta_{t-1}) - \nu_t$, where $\operatorname{Var}[\nu_t]=1$. The covariances of the term in parentheses are entries of $\bm{V}_{t-1}$. For the block structure, note the private signal errors $\eta_{i,t}$ are independent of events before $t$.} Now we have the following  formula for best-response actions, which, for simplicity, we give in the case where the agent has an improper prior\footnote{Our analysis extends immediately to any proper normal prior
for $\theta_{t}$: To get an agent's estimate of $\theta_{t}$, the
formula in (\ref{eq:action_formula}) would simply be averaged with
a constant term accounting for the prior, and everyone could invert
this deterministic operation to recover the same information from
others' actions.} about $\theta_t$.
\begin{equation}
a_{i,t}=\underbrace{\frac{\bm{1}^{\top}\bm{C}_{i,t-1}^{-1}}{\bm{1}^{\top}\bm{C}_{i,t-1}^{-1}\bm{1}}}_{\text{agent $(i,t)$'s weights}}\cdot\underbrace{\begin{pmatrix}\rho\bm{a}_{N_{i},t-1}\\
{\tiny \vdots}\\
\rho^{m}\bm{a}_{N_{i},t-m}\\
s_{i,t}
\end{pmatrix}}_{\text{agent $(i,t)$'s observations}}.\label{eq:action_formula}
\end{equation}
\noindent Expression
(\ref{eq:action_formula}) is a linear combination of the agent's
signal and the observed actions; the weights in this linear combination
depend on the matrix $\bm{V}_{t-1}$, but \emph{not} on realizations of any
random variables. Section \ref{OA-sec:details_updating} of the online Appendix gives the details of the standard calculations underlying the formula.

We denote by $(\boldsymbol{W}_{t},\boldsymbol{w}_{t}^{s})$ a \emph{weight
profile} in period $t$, with $\boldsymbol{w}_{t}^{s}\in\mathbb{R}^{n}$
being the weights agents place on their private signals and $\boldsymbol{W}_{t}$ being the weights they place on their other information.

\subsubsection{The evolution of covariance matrices under best-response behavior\label{subsec:evolution}}

Assuming agents best-respond according to the optimal weights just
described in (\ref{eq:action_formula}), we can compute the resulting
next-period covariance matrix $\bm{V}_{t}$ from the previous covariance
matrix. Letting $\mathcal{V}$ be the space of covariance matrices, this defines a map $\Phi:\mathcal{V}\to\mathcal{V}$, given
by 
\begin{equation}
\Phi:\bm{V}_{t-1}\mapsto\bm{V}_{t}.\label{eq:phi-def}
\end{equation}
This map gives the basic dynamics of the model: how an arbitrary variance-covariance
matrix $\bm{V}_{t-1}$ maps to a new one when all agents best-respond
to $\bm{V}_{t-1}$. The variance-covariance matrix $\bm{V}_{t-1}$
(along with parameters of the model) determines (i) the weights agents
place on their observations in (\ref{eq:action_formula}), and  (ii) the distributions of the random
variables that are being combined in this expression. This yields the deterministic updating
dynamic $\Phi$. A consequence is that the weights agents place on
observations are (commonly) known, and do not depend on any random
realizations.
\begin{example}
\label{exa:explicitPhi}We compute the map $\Phi$ explicitly in the
case $m=1$. We refer to the weight agent $(i,t)$ optimally places on $\rho a_{j,t-1}$
as $W_{ij}$ and the weight on $s_{i,t}$,
her private signal, as $w_{i}^{s}$. Note we have, from (\ref{eq:action_formula})
above, explicit expressions for these weights. Then 
\begin{equation}
\left[\Phi(\bm{V})\right]{}_{ii}=(w_{i}^{s})^{2}\sigma_{i}^{2}+\sum_{k,k'} W_{ik}W_{ik'}(\rho^{2}V_{kk'}+1)\quad\text{and}\quad\left[\Phi(\bm{V})\right]{}_{ij}=\sum_{k,k'} W_{ik}W_{jk'}(\rho^{2}V_{kk'}+1).\label{eq:variances-eqm-1}
\end{equation}
\end{example}

\subsection{Stationary equilibrium in linear strategies} \label{sec:Stationary_eqm_linear_strategies}

We will now turn our attention to \emph{stationary equilibria in linear
strategies}\textemdash ones in which all agents' strategies are linear
with time-invariant coefficients\textemdash though, of course, we
will allow all agents to consider deviating to arbitrary
strategies, including non-linear ones. Once we establish the existence
of such equilibria, we will use the word \emph{equilibrium} to refer
to one of these unless otherwise noted.

A reason for focusing on equilibria in linear strategies comes from
noting that, in the variant of the model with a starting time (i.e.,
the case $\mathcal{T}=\mathbb{Z}_{\geq0}$) agents begin by using
only private signals, and they do this linearly. After that, inductively
applying the reasoning of Section \ref{subsec:Best-Response-Weights},
best-responses are linear at all future times. Taking time to extend
infinitely backward is an idealization that allows us to focus on
exactly stationary behavior.

We now show the existence of stationary equilibria in linear strategies.
\begin{prop}
A stationary equilibrium in linear strategies exists, and is associated
with a covariance matrix $\widehat{\bm{V}}$ such that $\Phi(\widehat{\bm{V}})=\widehat{\bm{V}}$.
\label{prop:Existence}
\end{prop}
The proof appears in Appendix \ref{sec:Existence-of-equilibrium:appendix}, and we sketch the key ideas below.

At such an equilibrium, the covariance matrix $\bm{V}_{t}$ and all
agent strategies are time-invariant. Actions are linear combinations
of observations with stationary weights, which we denote by $\widehat{W}_{ij}$
and $\widehat{w}_{i}^{s}$. The form of these rules has some resemblance
to static equilibrium notions studied in the rational expectations
literature (e.g., \citealp*{vives1993fast,babus2018trading,lambert2018strategic,sleq}). It also has a similar form to the \cite{degroot1974reaching} and \cite{friedkin1997social} updating rules, typically imposed as behavioral heuristics. In our dynamic environment, such a solution emerges
as a stationary equilibrium.

\subsubsection{Proof sketch for the existence result}

The goal is to apply the Brouwer
fixed-point theorem to show there is a covariance matrix $\widehat{\bm{V}}$
that remains unchanged under updating. To find a convex, compact set
to which we can apply the fixed-point theorem, we use the fact that
when all agents best-respond to any beliefs about prior actions, all action variances lie in a compact set of positive numbers. This is because
all agents' actions must be at least as precise in estimating $\theta_{t}$
as their private signals, and cannot be more precise than estimates
given perfect knowledge of $\theta_{t-1}$ combined with the private
signal. This establishes bounds on action variances. The Cauchy-Schwartz inequality then bounds covariances
in terms of corresponding variances. All matrices respecting these bounds constitute a compact, convex set containing the image of $\Phi$. This and the
continuity of $\Phi$ allow us to apply the Brouwer fixed-point theorem.

\subsubsection{Other remarks}
In the case of $m=1$, we can use the formula of Example \ref{exa:explicitPhi},
equation (\ref{eq:variances-eqm-1}), to write the fixed-point condition
$\Phi(\widehat{\bm{V}})=\widehat{\bm{V}}$ explicitly. More generally,
for any $m$, equation (\ref{eq:action_formula}) gives a formula in terms of $\widehat{\bm{V}}$
for the weights $\widehat{W}_{ij}$ and $\widehat{w}_{i}^{s}$ in
the best response to $\widehat{\bm{V}}$, and this can be used to describe the
equilibrium $\widehat{V}_{ij}$ as solving a system of polynomial
equations. These equations typically have large degree and cannot be solved
analytically except in very simple cases, but they can readily be
used to solve for equilibria numerically. A related feature of the model is that
standard methods can easily be applied to estimate it and test
hypotheses within it (see Appendix \ref{OA-sec:Identification-and-Testable} for details).

The main insight is that we can analyze equilibria through action
covariances. This idea applies equally to many extensions and variations of our basic model, as illustrated by two examples: (1) We assume that agents observe neighbors
perfectly, but one could define other observation structures. For
instance, agents could observe actions with noise, or they could observe
some set of linear combinations of neighbors' actions with noise. Similarly, agents could be observing predecessors' actions for heterogeneous durations before acting (i.e., node-specific $m$).
(2) We assume agents are Bayesian and best-respond to the true
distribution of actions, but the same proof would also show that equilibria
exist under other behavioral rules (see Section \ref{subsec:Naive-agents}).\footnote{What is important in the proof is that actions depend continuously
on the covariance structure of an agent's observations; the action
variances are uniformly bounded under the rule agents play; and there
is a vanishing dependence of behavior on the very distant past.}

Proposition \ref{prop:Existence} shows that there exists a stationary
linear equilibrium. We show later, as part of Proposition \ref{prop:NondiverseSignals},
that there is a \emph{unique} stationary linear equilibrium in networks having
a particular structure. In general, uniqueness of the equilibrium
is an open question that we leave for future work.\footnote{We have checked numerically that $\Phi$ is not, in general, a contraction
in various norms (entrywise sup, Euclidean operator norm,
etc.). In computing equilibria numerically for many examples, we have
not been able to find a case of equilibrium multiplicity. Indeed,
in all of our numerical examples, repeatedly applying $\Phi$ to an
initial covariance matrix converges to the same fixed point for any starting
conditions.} In Section \ref{subsec:time-starts-subsec} and Appendix \ref{OA-sec:TimeZero},
we discuss the $\mathcal{T}=\mathbb{Z}_{\geq 0}$ variant of the model, which has a unique equilibrium,
and relate it to our main model.

How much information does each agent need to play her equilibrium
strategy? In a stationary equilibrium, she only needs to know the
steady-state variance-covariance matrix $\widehat{\bm{V}}_{N_{i}}$
in her neighborhood. Then her problem of inferring $\theta_{t-1}$
becomes essentially a linear regression problem. If historical empirical
data on neighbors' error variances and covariances are available, then $\widehat{\bm{V}}_{N_{i}}$
can be estimated from such data.

\section{How good is information aggregation in equilibrium?\label{sec:LearningOutcomes}}

In this section we analyze the quality of information aggregation
in stationary equilibrium. 

First, recall that in any agent's time-$t$ decision problem, $\theta_{t-1}$ is
a sufficient statistic for social information, because the difference $\theta_{t}-\rho \theta_{t-1}$ is independent of all actions taken at or before time $t-1$. Let the \emph{social signal} of agent $(i,t)$ be defined as her estimate of $\theta_{t-1}$ based on social information: $$r_{i,t}=\mathbb{E}[\theta_{t-1}\mid(\bm{a}_{N_{i},t-\ell})_{\ell=1}^{m}].$$

We will be interested in the error in this estimate:
\begin{defn} For a given strategy profile, define the \emph{aggregation error}
$\kappa_{i,t}^{2}=\text{Var}(r_{i,t}-\theta_{t-1})$ to be the expected squared error in
the social signal as a prediction of $\theta_{t-1}$. \end{defn} 
\noindent The aggregation error measures 
how well an agent can extract information from social observations. Note that agent $i$'s aggregation
error is a monotone transformation of her expected utility.\footnote{In fact, for any decision dependent on $\theta_{t}$, an agent is
better off with a lower value of $\kappa_{i,t}^{2}$. This is a consequence
of the fact that unidimensional Gaussian signals can be Blackwell
ordered by their precision.} 

How efficient is aggregation? The environment features informational externalities: players do not
internalize the impact of their learning rules on others' learning.
Consequently, there is no reason to expect outcomes to be  efficient in any exact sense. And we have seen that the details
of equilibrium in a particular network can be complicated. However,
it turns out that much more can be said about the behavior of aggregation
errors as neighborhood sizes 
grow. In this section, we study the asymptotic efficiency of information
aggregation. We give conditions under which aggregation error decays
as quickly as physically possible, and different conditions under
which it remains far from efficient levels even when agents have arbitrarily
many observations.  We discuss the case $m=1$ for simplicity but the reasoning extends
easily to other values of $m$.

\subsection*{A benchmark lower bound on aggregation error}

A first observation is a lower bound on the aggregation
error (in terms of an asymptotic rate as a function of a node's degree) under
\emph{any} behavior of agents. This establishes a benchmark relative
to which we can assess equilibrium outcomes.

Let $d_i$ denote the out-degree of a node $i$. 
\begin{fact}
\label{fact:bound}Fix $\rho\in(-1,1)$ as well as upper and lower
bounds for private signal variances, so that $\sigma_{i}^{2}\in[\underline{\sigma}^{2},\overline{\sigma}^{2}]$
for all $i$. On any network and for all strategy profiles,
we have $\kappa_{i,t}^{2}\geq c/d_{i}$ for all $i$ and $t$, where
$c$ is a constant that depends only on $\rho,\underline{\sigma}^{2},$
and $\overline{\sigma}^{2}$.
\end{fact}
The lower bound is reminiscent of the central limit theorem: if an
agent had $d_{i}$ conditionally independent noisy signals about $\theta_{t-1}$
(e.g., by observing neighbors' private signals directly), then the
variance of her estimate would be of order $1/d_{i}$. Fact \ref{fact:bound}
notes that it is not possible for aggregation
errors to decay (as a function of degree) any faster than that.

For an intuition, imagine that an agent sees neighbors' private signals
(not just actions) one period after they are received, and all other
private signals two periods after they are received; this clearly gives
an upper bound on the quality of the agent's possible aggregation given physical communication
constraints. The information that is two periods old cannot be very
informative about $\theta_{t-1}$ because of the movement in the state
from period $t-2$ to $t-1$; a large constant number $z$ of signals
about $\theta_{t-1}$ would be better. Thus, a lower bound on aggregation
error is given by the error that could be achieved with $d_{i}+z$
independent signals about $\theta_{t-1}$ of the best possible precision
($\underline{\sigma}^{-2}$). The bound follows from these observations.

\subsection*{Outline of results: When is aggregation comparable to the benchmark?}

Fact \ref{fact:bound} places a lower bound on aggregation error given
the physical constraints. Even efficient learning could not do better
than this bound. We examine when equilibrium learning can achieve
aggregation of similar quality. More precisely, we ask when there is
a stationary equilibrium where the aggregation error at
node $i$ satisfies $\widehat{\kappa}_{i}^{2}\leq C/d_{i}$ for all
$i$, for some constant $C$. 

In Section \ref{subsec:Diverse-Signals} we establish a good-aggregation
result: outcomes comparable to the benchmark are achieved in equilibrium in a class
of networks. The key condition enabling the asymptotically efficient
equilibrium outcome is called
\emph{signal diversity}: each individual has access to enough neighbors
with multiple different kinds of private signals. The fact that neighbors
use private information differently turns out to give the agents enough
power to identify $\theta_{t-1}$ with equilibrium aggregation error
that decays at a rate matching the lower bound of Fact 1 up to a multiplicative constant. 

In Section \ref{subsec:Non-diverse}, we turn to negative results.
Without signal diversity, equilibrium aggregation can be extremely
bad. Our first negative result shows that when signals are exchangeable,
it may be that the aggregation error $\widehat{\kappa}_{i}^{2}$ does not approach zero
in any equilibrium, no matter how large neighborhoods are, though a social planner could achieve good aggregation by prescribing different updating weights. We prove
this in highly symmetric networks. Once we move away from such networks,
one might ask whether diversity in individuals' network positions
could play a role analogous to signal diversity and enable approximately
efficient learning. Our next negative result shows that this is
impossible. When signals are homogeneous and all agents' degrees in network $G_n$ are bounded by $\overline{d}(n)$ (where $\overline{d}(n)$ is any unbounded sequence)  then in any equilibrium, it cannot be that almost all aggregation errors are less than $C/\overline{d}(n)$  as the network grows, for any number $C > 0$ not depending on $n$.

\subsection{Distributions of networks and signals\label{subsec:random-setting}}

For our good-aggregation result, we study large populations and specify
two aspects of the environment: \emph{network distributions }and \emph{signal
distributions}. In terms of network distributions, we work with a
standard type of random network model\textemdash a stochastic block
model (see, e.g., \citealp*{holland1983stochastic} and \citealp*{abbe2017community}). It makes the
structure of equilibrium tractable while also allowing us to capture
rich heterogeneity in network positions. We also specify \emph{signal
distributions}: how signal precisions are allocated to agents, in
a way that may depend on network position. We now formalize these two
primitives of the model and state the assumptions we work with.

Fix a set of \emph{network types} $k\in\mathcal{K}=\left\{ 1,2,\ldots,K\right\} $.
For each pair of network types,
there is a given probability $p_{kk'}$  that each agent of network type $k$ has a
link to each agent of network type $k'$. An assumption we maintain
on these probabilities is that each network type $k$ observes at
least one network type (possibly $k$ itself) with positive probability.
There is also a vector $(\alpha_{1},\ldots,\alpha_{K})$ of \emph{population
shares} of each network type, which we assume are all positive. Jointly, $(p_{kk'})_{k,k'\in\mathcal{K}}$
and $\bm{\alpha}$ specify the network distribution. These parameters
can encode differences in expected degree and also features such as
homophily (where some groups of types are linked to each other more
densely than to others). 

We next define signal distributions, which describe the allocation
of signal variances to network types. Fix a finite set $\mathbb{S}$
of private signal variances, which we call signal types.\footnote{The assumptions of finitely many signal types and network types are
for technical convenience only, and could be relaxed.} We let $q_{k\tau}$ be the share of agents of network type $k$ with
signal type $\tau$; then $(q_{k\tau})_{k\in\mathcal{K},\tau\in\mathbb{S}}$
defines the signal distribution.

Let the nodes in network $n$ be partitioned into the network types $N_{n}^{1},N_{n}^{2},\ldots,N_{n}^{K}$,
with the cardinality $|N_{n}^{k}|$ equal to $\lfloor\alpha_{k}n\rfloor$
or $\lceil\alpha_{k}n\rceil$ (rounding so that there are $n$ agents
in the network). We (deterministically) set the signal variances $\sigma_{i}^{2}$
equal to elements of $\mathbb{S}$ in accordance with the signal shares
(again rounding as needed). Let $\left(G_{n}\right)_{n=1}^{\infty}$
be a sequence of directed or undirected random networks with these
nodes, so that $i\in N_{n}^{k}$ and $j\in N_{n}^{k'}$ are linked
with probability $p_{kk'}$; these link realizations are all independent.

In our setting, a  \emph{stochastic block model} $D$ is specified by the linking probabilities
$(p_{kk'})_{k,k'\in\mathcal{K}}$, the type shares $\bm{\alpha}$,
and the signal distribution $(q_{k\tau})_{k\in\mathcal{K},\tau\in\mathbb{S}}$.
We let $(G_{n}(D),\bm{\sigma}_{n}(D))$ denote the environment (i.e., the
network and the signal variances) in a random realization. We say that
 a network type $k$ contains a signal type $\tau$ if $q_{k\tau}>0$.
\begin{defn}
A stochastic block model satisfies \emph{signal diversity}
if each network type has a positive probability of linking with at
least one network type containing two distinct signal types.
\end{defn}

\subsection{Good aggregation under diverse signals\label{subsec:Diverse-Signals}}

Our first main result is that signal diversity is sufficient for good
aggregation in the networks described in the previous section. Aggregation
error decays at a rate $C/d_{i}$ for each node $i$  independently of
the structural properties of the network.

We first define a notion of good aggregation for an agent in terms
of a bound on that agent's aggregation error.
\begin{defn}
Given $\varepsilon>0$, we say that agent $i$ achieves the $\varepsilon$-aggregation
benchmark in a given equilibrium if the aggregation error satisfies $\widehat{\kappa}_{i}^{2}\leq\varepsilon$.
\end{defn}
We say an event (indexed by $n$) occurs \emph{asymptotically almost surely} if the probability of the event converges to $1$ as $n\rightarrow \infty$.
\begin{thm}
\label{thm:DiverseSignals}Fix any stochastic block model $D$ satisfying
signal diversity. There exists $C>0$ such that asymptotically almost
surely the environment $(G_{n}(D),\bm{\sigma}_{n}(D))$ has an equilibrium where all agents achieve the $C/n$-aggregation benchmark.
\end{thm}
So for large enough $n$, societies with signal diversity are very likely to aggregate information
very well. The uncertainty in this statement is over the network,
as there is always a small probability of a realized network which
prevents learning (e.g., an agent has no neighbors). We give an outline
of the argument next, and the proof appears in Appendix \ref{sec:Proof-of-Theorem-diverse}.

The constant $C$ in the theorem statement can depend on the stochastic block model $D$. However, given any compact set of stochastic block models $D$,
we can choose a single $C>0$ for which the result holds uniformly
across $D$.\footnote{The reason is that the distribution of aggregation errors is upper hemicontinuous in model parameters, so if the desired bounds hold for each point in a compact set, they can be made uniform.} Thus, the theorem can be applied without detailed information
on how the random graphs are generated, as long as some bounds are
known about which models are possible.

\subsubsection{Discussion of the proof} \label{sec:discussion_diversity_proof}

To give intuition for Theorem \ref{thm:DiverseSignals}, we first
describe why the theorem holds on the complete network\footnote{Note this is a special case of the stochastic block model.}
with two signal types $A$ and $B$ in the $m=1$ case. This echoes
the intuition of the example in the introduction. We then discuss
the challenges involved in generalizing the result to our general
stochastic block model networks, and the techniques we use to overcome
those challenges.

Consider a time-$t$ agent, $(i,t)$. Recall that the social signal 
$r_{i,t}$ is the optimal estimate of $\theta_{t-1}$ based on
the actions $(i,t)$ has observed in her neighborhood. In the complete
network, all agents have the same social signal, which we call $r_{t}$.\footnote{In particular, agent $(i,t)$ sees everyone's past action, including the one taken last period at the same node.}

At any equilibrium, each agent's action is a weighted average of her
private signal and this social signal.
\begin{equation}
a_{i,t}=\widehat{w}_{i}^{s}s_{i,t}+(1-\widehat{w}_{i}^{s})r_{t}.\label{eq:asr}
\end{equation}
The weights on the two random variables on the right-hand side sum to $1$ because  both $s_{i,t}$ and $r_t$ are unbiased estimates of $\theta_t$, and so is the left-hand side $a_{i,t}$.  The weight $\widehat{w}_{i}^{s}$ on the private signal depends on the precision of this signal relative to the social signal. We call the weights used by agents of the two
distinct signal types $\widehat{w}_{A}^{s}$ and $\widehat{w}_{B}^{s}$.
Suppose signal type $A$ is more precise than signal type $B$, so
that $\widehat{w}_{A}^{s}>\widehat{w}_{B}^{s}$.

Now, turning our attention to the next period of updating, observe that each time-$(t+1)$ agent can compute two averages
of the time-$t$ actions\textemdash one for each signal type. Using (\ref{eq:asr}) to rewrite $a_{i,t}$ and then plugging in $s_{i,t}=\theta_{t}+\eta_{i,t}$:
\setlength{\fboxsep}{0pt}
\setlength{\fboxrule}{0pt}
\[
\text{\fbox{%
    \parbox{1in}{\footnotesize \raggedright type $A$ average action at time $t$}}} = \frac{1}{n_{A}}\;\sum_{i:\sigma_{i}^{2}=\sigma_{A}^{2}}a_{i,t}=\widehat{w}_{A}^{s}\theta_{t}+(1-\widehat{w}_{A}^{s})r_{t}+O_p(n^{-1/2}),
\]
\[
\text{\fbox{%
    \parbox{1in}{\footnotesize \raggedright type $B$ average action at time $t$}}} = \frac{1}{n_{B}}\;\sum_{i:\sigma_{i}^{2}=\sigma_{B}^{2}}a_{i,t}=\widehat{w}_{B}^{s}\theta_{t}+(1-\widehat{w}_{B}^{s})r_{t}+O_p(n^{-1/2}).
\]
Here $n_{A}$ and $n_{B}$ denote the numbers of agents of each type (recalling we assumed each type is a positive share of the population size, $n$). 
The $O_p(n^{-1/2})$ error terms\footnote{The notation means the errors are bounded by  $Cn^{-1/2}$ for a $C>0$ with high probability \citep{janson}.} come from the average signal noises
$\eta_{i,t}$ of agents in each group; the  bound holds with high probability by the central
limit theorem. In other words, each time-($t+1$)
agent can obtain precise estimates of two different convex combinations
of $\theta_{t}$ and $r_{t}$. Because the two weights, $\widehat{w}_{A}^{s}$
and $\widehat{w}_{B}^{s}$, are distinct, she can approximately  (up to signal error) solve
for $\theta_{t}$ as a linear combination of the average actions taken by
each type she observes. It follows the agent must have an estimate at least as precise
as what she can obtain by the strategy we have described, and will
thus be very close the benchmark. Since the equilibrium in question was arbitrary, this shows that aggregation approaches the benchmark in any equilibrium. The estimator of $\theta_{t}$ in
this strategy places negative weight on $\frac{1}{n_{B}}\sum_{i:\sigma_{i}^{2}=\sigma_{B}^{2}}a_{i,t}$,
thus \emph{anti-imitating} the agents of signal type B---those with the less precise private signal. The logic of Proposition \ref{prop:asymundir} in Section \ref{sec:anti-imitation-necessary} implies that anti-imitation necessarily occurs  in any equilibrium in which agents aggregate information precisely.

To extend the ideas just presented to the more general setting of Theorem \ref{thm:DiverseSignals}, we need to show that each individual
observes a large number of neighbors of at least two signal types who also have similar
social signals. More precisely, the proof shows that agents with the
same network type have highly correlated social signals. Showing this
is much more subtle than it was in the above illustration. In general,
the social signals in an arbitrary network realization are endogenous
objects that depend to some extent on all the links.

A key insight allowing us to overcome this difficulty is a useful general fact about sufficiently dense stochastic block models: despite
a lot of idiosyncratic randomness in direct connections,
the law of large numbers implies  the number of paths of length two between any agent $i$ of type $k$ and any agent $j$ of type $k'$ going through an agent of type $k''$ is
nearly determined by the types $k$, $k'$, and $k''$, with a small relative error.\footnote{
For simplicity we first present the argument in a random graph family where the number of two-step paths is nearly deterministic. The argument extends to a larger class of models where the same property applies to longer paths, as we discuss in the next subsection.} We can leverage this to deduce some important facts about the updating map $\Phi$ (recall Section \ref{subsec:evolution})
in the realized random network, and specifically
about the evolution of social signals.

In particular, if we look at the
set of covariance matrices where all social signals are close to perfect,
we can show that the composition $\Phi^2 := \Phi \circ \Phi$ maps this set to itself. In other words, if social signals are very precise, then they will remain very precise two periods later. If the two-step path counts were determined by types exactly, it would not be too difficult to show this by elaborating the reasoning in the complete graph example, because neighbors of the same type would be effectively exchangeable. We show that despite the path counts being known only approximately, the desired conclusion holds. This is nontrivial because the weights agents use in
their updating\textemdash and thus the evolution of social signals\textemdash could depend sensitively
on realized network structure; small relative errors could matter.  A key step is to develop
results on matrix perturbations to show that small relative changes
in the network actually do not affect $\Phi^{2}$ too much. A
fixed-point theorem then implies there is a fixed point of $\Phi^{2}$
in the set of outcomes with very precise social signals. With some further analysis we can deduce that this implies the existence of an equilibrium (corresponding to a fixed point of $\Phi$)
with nearly perfect aggregation.

\subsubsection{Sparser random graphs}

In the random graphs we have defined in Section \ref{subsec:random-setting}, the group-level linking probabilities
$(p_{kk'})$ are, for simplicity, held fixed as $n$ grows. This yields
expected degrees that grow linearly in the population size, which
may not be the desired asymptotic model. We can, however, establish versions of our results in a class of models much more flexible with respect to degrees. While it is important to have neighborhoods ``large enough'' (i.e., growing in $n$) to permit
the application of laws of large numbers, their rate of growth can
be considerably slower than linear. For example, our proof can be extended directly to degrees that scale as $n^{\alpha}$ for any $\alpha>0$ to show that asymptotically almost surely, there exists an equilibrium
where the $C/n^{\alpha}$-aggregation benchmark is achieved for all
agents.
Instead of studying $\Phi^{2}$ and two-step paths, one can 
extend the same sort of analysis to the $L$-fold composition $\Phi^{L}$, which reflects $L$-step paths.  In order to do this, one uses the fact that for $L$ larger
than $1/\alpha$, the number of paths of length $L$ between any two nodes is determined by the types involved in the path with a small relative error. Elaborating the proof of the theorem above, we can then characterize the behavior of $\Phi^L$ and finally deduce the claimed aggregation property for $\Phi$.

\subsubsection{The good-aggregation outcome as a unique prediction\label{subsec:time-starts-subsec}}

The theorem above says good aggregation is supported in an equilibrium
but does not state that this is the unique equilibrium outcome. To
deal with this issue, we study the alternative model with $\mathcal{T}=\mathbb{Z}_{\geq0}$
(where agents begin with only their own signals and then best-respond
to the previous distribution of behavior at each time). We show that, as  $n\to\infty$,
its long-run outcomes get arbitrarily close to the good-aggregation
equilibrium of Theorem \ref{thm:DiverseSignals}  under the same conditions.
Thus, even if there were other equilibria of the stationary model,
they could not be approached via the natural iterative procedure coming
from the $\mathcal{T}=\mathbb{Z}_{\geq0}$ model. Formal statements and details are in Appendix \ref{OA-sec:TimeZero}.

\subsection{Aggregation under homogeneous signals\label{subsec:Non-diverse}}

Having established conditions for good aggregation under signal diversity,
we now explore what happens without signal diversity. Our general
message is that aggregation is worse.

To gain an intuition for this, note that it is essential to the argument
described in the previous subsection that different agents have different
signal precisions. Recall the complete network case. From the perspective
of an agent $(i,t+1)$, the fact that type $A$ and type $B$ neighbors
place different weights on the social signal $r_{t}$ keeps their behavior from being collinear, and allows $(i,t)$ to separate $\theta_{t}$ from a
confound. In that example, if type $A$ and $B$ agents had the same signal types, they would use the same weights, and our agent trying to learn from them would face a collinearity problem.

We begin by studying graphs having a symmetric structure and show
that learning outcomes are necessarily bounded very far from good
aggregation. We then turn to arbitrary large graphs and prove a lower bound
on aggregation error that implies the homogeneous-signals regime has,
quite generally, worse outcomes for some agents than those achieved
by everyone in our good-aggregation result.

\subsubsection{Aggregation in networks with symmetric neighbors\label{subsec:Graphs-with-symmetric}}
\begin{defn}
A network $G$ has \emph{symmetric neighbors} if, whenever  $j,j'\in N_{i}$ for some $i$, then $N_{j}=N_{j'}$.
\end{defn}
\noindent In the undirected case, the graphs with symmetric neighbors are the
complete network and complete bipartite networks.\footnote{These are both special cases of our stochastic block model from Section
\ref{subsec:Diverse-Signals}, so Theorem  \ref{thm:DiverseSignals} applies to these network structures when signal diversity is satisfied.} For directed graphs, the condition allows a larger variety of networks.

\begin{prop}
\label{prop:NondiverseSignals}Consider a sequence $\left(G_{n}\right)_{n=1}^{\infty}$ of strongly
connected graphs with symmetric neighbors. Assume that all signal
variances are equal, and that $m=1$. Then there is a unique equilibrium on each $G_n$, and there exists an $\varepsilon>0$ 
such that the $\varepsilon$-aggregation benchmark is not achieved
by any agent $i$ at this equilibrium for any $n$.
\end{prop}
All agents have non-vanishing aggregation errors at the unique equilibrium.
So all agents learn poorly compared to the diverse signals case. The
proof of this proposition, and the proofs of all subsequent results,
appear in Appendix \ref{OA-sec:remaining-proofs}.

This failure of good aggregation is not due simply to a lack of sufficient
information in the environment: On the complete graph with exchangeable
(i.e., non-diverse) signals, a social planner who set
weights for all agents could achieve $\varepsilon$-aggregation for
any $\varepsilon>0$ when $n$ is large. See Appendix \ref{OA-sec:Socially-optimal-learning}
for a formal statement, proof and numerical results.\footnote{We thank Alireza Tahbaz-Salehi for suggesting this analysis.}
In this sense, the social learning externalities are quite severe:
a small change in weights for each individual could yield a
very large benefit in a world of homogeneous signal types.

We now give intuition for Proposition~\ref{prop:NondiverseSignals}.
In a graph with symmetric neighbors and homogeneous signals, in the unique equilibrium,\footnote{The proof of the proposition establishes uniqueness by showing that
$\Phi$ is a contraction in a suitable sense.}
actions of any agent's neighbors are exchangeable. So Bayesian estimates (and thus actions) must weight all neighbors equally, which prevents
the sort of inference of the most recent state illustrated in Section \ref{sec:discussion_diversity_proof}. This is easiest to see on the complete graph, where \emph{all}
observations are exchangeable. So, in any equilibrium, each agent's
action at time $t$ is equal to a weighted average of her own signal
and the average action $\frac{1}{|N_{i}|}\sum_{j\in N_{i}}a_{j,t-1}$: 
\begin{equation}
a_{i,t}=\widehat{w}_{i}^{s}s_{i,t}+(1-\widehat{w}_{i}^{s})\frac{1}{|N_{i}|}\sum_{j\in N_{i}}a_{j,t-1}.\label{eq:asr-1}
\end{equation}
By iteratively using this equation, we can see that actions must place
substantial weight on the average of signals from, e.g., two periods
ago, and indeed further back. Note that all signals $s_{j,t'}$ at past times $t'$ take the form $\theta_{t'}+\eta_{i,t'}$. Thus, although the effect of \emph{signal errors} $\eta_{i,t'}$
vanishes (by averaging) as $n$ grows large, the correlated error from \emph{past
changes in the state} $\nu_{t'}$ never ``washes out'' of estimates, and this
is what prevents vanishing aggregation errors.

The bad-aggregation result as stated applies to exactly homogeneous signal types only. In fact, in finite networks we need sufficiently heterogeneous signals to avoid bad learning outcomes; this is illustrated in Appendix \ref{OA-subsec:Villages}. In Section \ref{sec:welfare} we discuss the welfare implications of this failure of aggregation.

As a consequence of Theorem \ref{thm:DiverseSignals} and Proposition
\ref{prop:NondiverseSignals}, we can give an example where making
one node's private information less precise helps all agents.
\begin{cor}
\label{cor:InfoCompStat}There exists a network $G$, a vector of signal precisions $\bm{\sigma}$, and an agent $i\in G$
such that increasing $\sigma_{i}^{2}$ yields a Pareto improvement
at the unique equilibrium.
\end{cor}
To prove the corollary, we consider the complete graph with homogeneous
signals and large $n$. By Proposition \ref{prop:NondiverseSignals},
all agents have non-vanishing aggregation errors. If we instead give
agent $1$ a very uninformative signal, all players can anti-imitate
agent $1$ and achieve vanishing aggregation errors. When the signals
at the initial configuration are sufficiently imprecise, this gives
a Pareto improvement. There are also examples where severing links in the observational network can yield a Pareto improvement, as reported in an earlier version of the present paper \citep*{WPversion}.

\subsubsection{Aggregation in arbitrary networks}

Section \ref{subsec:Graphs-with-symmetric} showed aggregation errors
are non-vanishing when signal endowments and neighborhoods are symmetric.
A natural question is whether asymmetry in network positions can substitute
for asymmetry in signal endowments. In Section \ref{subsec:Diverse-Signals}
the key point was that different neighbors' actions were informative
about different linear combinations of $\theta_{t}$ and older information,
and this permitted filtering. Perhaps different network positions
can achieve the same effect?

We thus move to \emph{arbitrary} networks and show a weaker but much
more general result. Consider any sequence of equilibria on any networks
with symmetric signal endowments. Our result here is that no equilibrium
achieves $C/n$-aggregation for almost all agents, no matter what
$C$ is. In particular, this implies that the rate of learning (as $n$ grows) is
slower than at the good-learning equilibrium with diversity of signal
endowments from Theorem \ref{thm:DiverseSignals}. Moreover, if degrees
are bounded above by some $\overline{d}(n)$ growing at rate slower than $n$,
we prove the stronger statement that no equilibrium achieves $C/\overline{d}(n)$-aggregation
for almost all agents.
\begin{thm}
\label{thm:NonDiverseSignalsArb}Let $C>0$. Let $(G_{n})_{n=1}^{\infty}$
be an arbitrary sequence of networks and suppose all private signals
have variance $\sigma^{2}$. If all agents' in-degrees and out-degrees
are bounded above by some $\overline{d}(n)\rightarrow\infty$, then
in any sequence of equilibria, the aggregation error $\widehat{\kappa}_{i}^{2}$ is greater than $C/\overline{d}(n)$
for a non-vanishing fraction of agents $i$.
\end{thm}
In addition to considering arbitrary networks, we allow the memory
$m$ to be an arbitrary positive integer. Because the assumptions are much
weaker, we obtain a weaker conclusion than in Proposition \ref{prop:NondiverseSignals}.
While Proposition \ref{prop:NondiverseSignals} shows that aggregation
errors are non-vanishing,  this theorem shows that aggregation errors
cannot vanish quickly, but does not rule out aggregation errors vanishing
more slowly.

The basic intuition is that to avoid putting substantial weight on
$\theta_{t-2},$ an agent at time $t$ must anti-imitate some neighbors.
If all or almost all neighbors achieve $C/n$-aggregation for some
$C$ and have identical types of private signals, there is not much diversity among neighbors. So more and more
anti-imitation is needed as $n$ grows large in the sense that the
total positive weight and total negative weight on neighbors both
grow large. But then the contribution to the agent's variance from
neighbors' private signal errors cannot vanish quickly.

We can combine Theorems \ref{thm:DiverseSignals} and \ref{thm:NonDiverseSignalsArb}
to compare the value of signal diversity and network diversity. With
diversity of signal endowments, there exists a $C>0$ such that asymptotically
almost surely there is a good-learning equilibrium achieving the $C/n$-aggregation
benchmark for all agents under the stochastic block model. With exchangeable
signals, it is not possible to find equilibria achieving the same aggregation rate in $n$ under \emph{any} sequence of networks. Thus,  Theorem \ref{thm:NonDiverseSignalsArb} shows that network heterogeneity cannot improve learning outcomes as much as signal heterogeneity.  Section \ref{sec:village_numerics_body} complements the asymptotic results with numerical results in finite networks. It shows that in our model on real-world (highly asymmetric) social networks, signal heterogeneity improves learning outcomes much more than choosing a very favorable network structure but homogeneous signals.

\subsection{The welfare loss associated with homogeneity}\label{sec:welfare}

The results derived so far in this section show that there is a qualitative difference in how well agents are able to infer recent states across the homogeneous and heterogeneous signal settings. How important is this difference for welfare? We illustrate next that the welfare loss associated with signal homogeneity can be arbitrarily severe.

To gain an intuition for this, note that with homogeneous signals, period-$t$ actions are confounded by previous states. These confounds include $\theta_{t-2}$, which all $t-1$ agents use in the same way (as illustrated in the example of the introduction). But the confounds also include $\theta_{t-3}$, which could not be filtered out by $t-1$ agents, and so forth. The more weight agents place on social information (i.e., the more informative the past is), the more severe this confounding is. If the state is highly persistent and private signals are not very precise, then the confounds from periods even very long ago are substantial. The following corollary quantifies this effect.

\begin{cor} \label{cor:welfare_loss}
Consider a complete graph with all signal variances equal to $\sigma^{2}$, and let $m=1$. Then, in any symmetric strategy profile,
$$ \text{Var}(a_{i,t}-\theta_t)  \geq \frac{(1-\widehat{w}^s)^2}{1-(1-\widehat{w}^s)^2\rho^2},$$ where $\widehat{w}^s$ is the weight agents place on their own signals. As  $\rho \to 1$ from below and $\sigma^{-2}\to 0$, agent $i$'s action error in the unique equilibrium 
tends to infinity. Moreover, this convergence is uniform in $n$.\footnote{For any $\underline{v}$, there are $\underline{\rho}<1$ and $\overline{\sigma}^{-2}>0$ such that if $\rho >\underline{\rho}$ and $\sigma^{-2}<\overline{\sigma}^{-2}$, then $\text{Var}(a_{i,t}-\theta_t)  \geq \underline{v}$ for all $n$. } 
\end{cor}

The corollary guarantees that we can choose $(\sigma^{-2},\rho)$ so that the error is arbitrarily large, uniformly in $n$. In contrast, recall that our main positive result shows that the $C/n$-aggregation benchmark would be achieved with signal heterogeneity.\footnote{For example, by making half the agents' signals strictly worse.} When this benchmark is achieved, each individual obtains a variance $\text{Var}(a_{i,t}-\theta_t)$  that is at worst $1$ if $n$ is large enough.\footnote{Note the agent can use the estimate of last period's state, which has an error of order $C/n$. If the agent simply set her action equal to this estimate, then she would achieve $\text{Var}(\theta_{t}-\theta_{t-1})= 1$, since the state innovation has variance $1$. Additionally using her private signal does strictly better than this.} This bound on variance  does not depend on $\sigma^2$ or $\rho$. Thus \emph{welfare can be arbitrarily worse in environments with signal homogeneity compared to ones with heterogeneity}.

In large complete graphs with homogeneous signals, we can explicitly characterize the limit action variance (and therefore welfare). Let $\Vtext^{\infty}$ denote the limit, as $n$ grows large,
of  $\text{Var}(a_{i,t}-\theta_{t})$. Let
$\Covar^{\infty}$ denote the limit covariance of any two
agents' errors. By direct computation using equation (\ref{eq:variances-eqm-1}), these can be seen to be related
by the following equations, which have a unique solution: 
\begin{equation}
\Vtext^{\infty}=\frac{1}{\sigma^{-2}+(\rho^{2}\Covar^{\infty}+1)^{-1}},\qquad\qquad \Covar^{\infty}=\frac{(\rho^{2} \Covar^{\infty}+1)^{-1}}{[\sigma^{-2}+(\rho^{2}\Covar^{\infty}+1)^{-1}]^{2}}.\label{eq:constants}
\end{equation}

These equations also let us extend Corollary \ref{cor:welfare_loss} beyond the complete graph. The $\Vtext^{\infty}$ and $\Covar^{\infty}$ solving (\ref{eq:constants}) describe the limits of all variances and covariances in any graph with symmetric neighbors where degrees tend
uniformly to infinity.\footnote{Indeed, it can be deduced (as in the proof of Corollary~\ref{cor:welfare_loss}) that agents' actions are equal
to an appropriately discounted sum of past $\theta_{t-\ell}$, up
to error terms (arising from $\eta_{i,t-\ell}$) that vanish asymptotically. The weights on past states are the same as in the complete-network case, which is why the characterization of (\ref{eq:constants}) applies.} As $\sigma^{-2} \rightarrow 0$ and $\rho \rightarrow 1$ from below, equations (\ref{eq:constants}) show that $\Vtext^{\infty}$  and therefore $\Covar^{\infty}$ diverge to infinity, just as in the complete-network case. This shows the welfare loss from homogeneity can also be arbitrarily severe in graphs with symmetric neighbors and large degrees.

\section{The importance of understanding correlations\label{sec:Importance-Anti-Im}}

In the positive result on achieving the $C/n$-aggregation benchmark
(Theorem \ref{thm:DiverseSignals}), a key aspect of the argument
involved agents filtering out confounding information from their neighbors'
estimates\textemdash i.e., responding in a sophisticated way to the
correlation structure of those estimates. In this section, we demonstrate
that this sort of behavior is essential for nearly perfect aggregation,
and that more naively imitative heuristics yield outcomes far from
the benchmark. Empirical studies have found evidence (depending on
the setting and the subjects) consistent with both equilibrium behavior
and naive inference in the presence of correlated observations (e.g.,
\citealp*{eyster2015experiment,dasaratha2019experiment,enke2013correlation}).

We begin with a canonical model of agents who do not account for correlations
among their neighbors' estimates conditional on the state, and show
by example that naive agents achieve much worse learning than Bayesian
agents, and thus have non-vanishing aggregation errors. We then formalize
the idea that accounting for correlations in neighbors' actions is
crucial to reaching the benchmark. This is done by demonstrating a
general lack of good aggregation by agents who use imitative strategies,
rather than filtering in a sophisticated way. Finally, we show that
even in fixed, finite networks, any positive weights chosen by optimizing
agents will be Pareto-dominated.

\subsection{Naive agents\label{subsec:Naive-agents}}

In this part we introduce agents who misunderstand the distribution
of the signals they are facing and who therefore do not update as
Bayesians with a correct understanding of their environment. We consider
a particular form of misspecification that simplifies solving for
equilibria analytically:\footnote{There are a number of possible variants of our behavioral assumption,
and it is straightforward to numerically study alternative specifications
of behavior in our model (\citealp{alatas2016network} consider one
such variant).} 
\begin{defn}
We call an agent\emph{ naive} if she believes that all neighbors choose
actions equal to their private signals and maximizes her expected
utility given these incorrect beliefs.
\end{defn}
Equivalently, a naive agent believes her neighbors all have empty
neighborhoods. This is the analogue, in our model, of ``best-response
trailing naive inference'' \citep{eyster2010naive}. So naive agents
understand that their neighbors' actions from the previous period
are estimates of $\theta_{t-1}$, but they think these are conditionally independent given the state, and that the precision of each estimate
is equal to the signal precision of the corresponding agent. They
then play their expectation of the state given this misspecified theory
of others' play.

In Figure \ref{fig:Bayesian-and-Naive}, we compare Bayesian and naive
learning outcomes. We consider a complete network with 600 agents
and $\rho=0.9$. Half of agents have signal variance $\sigma_{A}^{2}=2$,
while we vary the signal variance $\sigma_{B}^{2}$ of the remaining
agents. The figure shows the average social signal variance for the
group of agents with private signal variance $\sigma_{A}^{2}=2$.
It suggests that naive agents learn substantially worse than
rational agents, whether signals are diverse or not. We prove this
holds for general stochastic block models and provide formulas for
variances under naive learning in Appendix \ref{OA-sec:Naive-Appendix}.

\begin{figure}

\centering{}\includegraphics[scale=0.4,trim={0cm 5.5cm 0cm 5.5cm},clip]{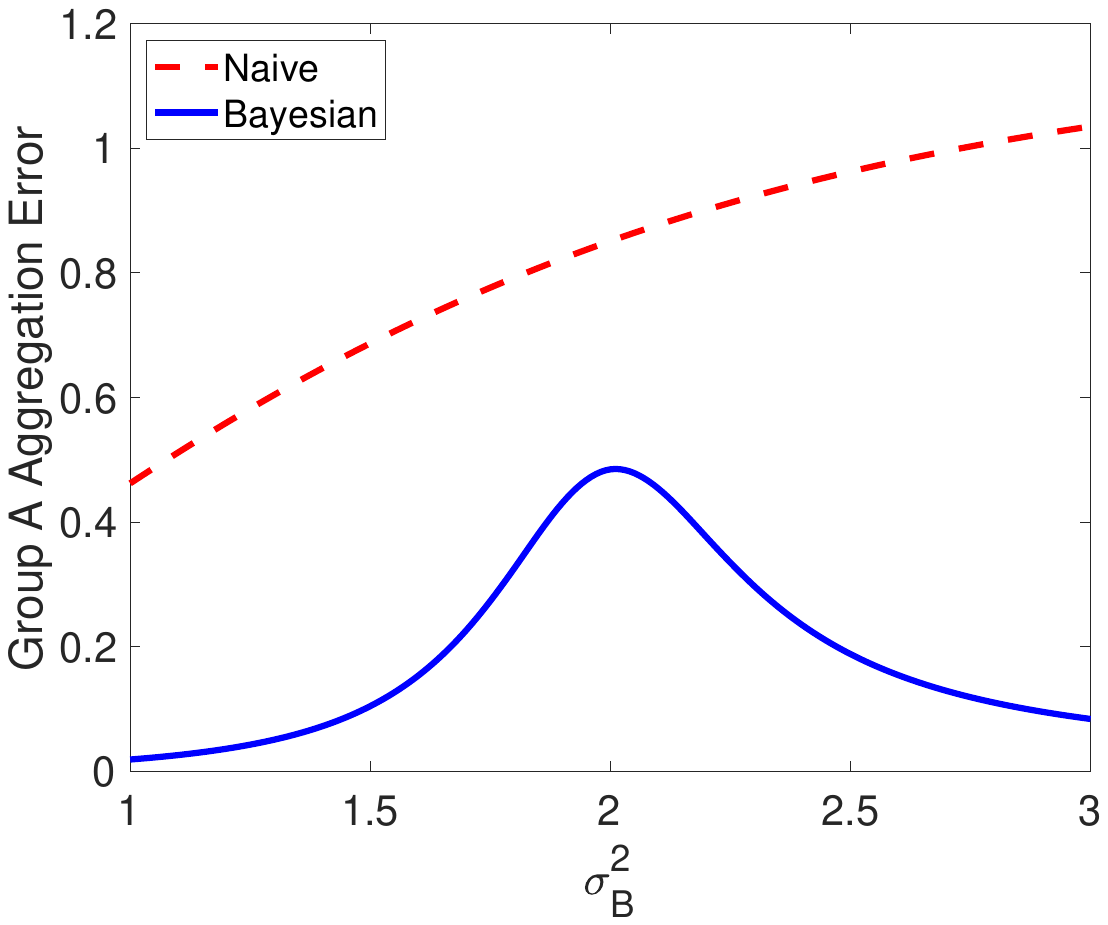}

\caption{\footnotesize Bayesian and naive learning on a complete graph and $n=600$ agents divided into two groups of equal size. The plot shows the aggregation error in group A as group B's private signal variance varies, fixing group $A$'s private signal variance at $\sigma_A^2=2$.}
\label{fig:Bayesian-and-Naive}
\end{figure}

\subsection{More general learning rules: Understanding correlation is essential
for good aggregation} \label{sec:anti-imitation-necessary}

We now show more generally that a sophisticated response to correlation
is needed to achieve vanishing aggregation errors on any sequence
of growing networks. To this end, we make the following definition:
\begin{defn}
The $\textit{steady state}$ associated with weights $\bm{W}$ and
$\bm{w}^{s}$ is the (unique) covariance matrix $\bm{V}^{*}$ such
that if actions have a variance-covariance matrix given by $\bm{V}_{t}=\bm{V}^{*}$
and next-period actions are set using weights $(\bm{W},\bm{w}^{s})$,
then $\bm{V}_{t+1}=\bm{V}^{*}$ as well.
\end{defn}
In this definition of steady state, instead of best-responding to
others' actual distributions of play, agents use exogenous weights
$\bm{W}$ in all periods.

By a straightforward application of the contraction mapping theorem,
if agents use any non-negative weights under which covariances remain
bounded at all times, there is a unique steady state.

Consider a sequence of networks $\left(G_{n}\right)_{n=1}^{\infty}$
with $n$ agents in $G_{n}$.
\begin{prop}
\label{prop:asymundir} Fix any sequence $\left(\bm{V}^{*}(n)\right)_{n=1}^{\infty}$, with each $\bm{V}^{*}(n)$ being a steady state under non-negative
weights in the network $G_{n}$. Suppose that all private signal variances are
bounded below by $\underline{\sigma}^{2}>0$ and that all agents place
weight at most $\overline{w}<1$ on their private signals. Then there
is an $\varepsilon>0$ such that, for all $n$, the $\varepsilon$-aggregation
benchmark is not achieved by any agent $i$ at the steady state $\bm{V}^{*}(n)$.
\end{prop}
The essential idea is that at time $t+1$, observed time-$t$ actions
all put weight on actions from period $t-1$, which causes $\theta_{t-1}$
to have a (positive weight) contribution to all observed actions.
Agents do not know $\theta_{t-1}$ and, with positive weights, cannot
take any linear combination that would recover it. Even with a very
large number of observations, this confound prevents agents from learning
the time-$t$ state precisely.

We now explain why we impose an assumption of all weights on private signals being bounded away from
$1$. If there were many autarkic agents who simply reported
their private signals (i.e., placed weight $1$ on these signals), some other agent could learn well without adjusting
for correlations by observing the autarkic agents. Note that in this case, all of the autarkic agents
would have non-vanishing aggregation errors. This illustrates that a weaker conclusion than that of the proposition can be established more generally.  If we did not impose a bound on
private signal weights, learning would fail in the weaker sense that \emph{some} agent must fail to achieve
the $\varepsilon$-aggregation benchmark for small enough $\varepsilon$.

On undirected networks, the proposition implies that aggregation errors
do not vanish under naive inference or under various other specifications
of non-Bayesian inference implying nonnegative weights. Moreover, the same argument shows that
in any sequence of \emph{Bayesian} equilibria on undirected networks
where all agents use positive weights, no agent can learn well.

\subsection{Without anti-imitation, outcomes are Pareto-inefficient\label{subsec:pareto}}

The previous section argued that anti-imitation is critical to achieving
vanishing aggregation errors. We now show that even in small networks,
where that benchmark is not relevant, any equilibrium without anti-imitation
is Pareto-inefficient relative to another steady state. This result
complements our asymptotic analysis by showing a different sense (relevant
for small networks) in which anti-imitation is necessary to make the
best use of information.
\begin{prop}
\label{thm:Pareto}Suppose the network $G$ is strongly connected
and some agent has more than one neighbor. Given any naive equilibrium
or any Bayesian equilibrium where all weights are positive, the action
variances at that equilibrium are Pareto-dominated by action variances
at another steady state.
\end{prop}
The basic argument behind Proposition \ref{thm:Pareto} is that if
agents place marginally more weight on their private signals, this
introduces more independent information that eventually benefits everyone. In a review of sequential learning experiments, \citet*{weizsacker2010we}
finds that subjects weight their private signals more heavily than
is optimal (given the empirical behavior of others they observe).
Proposition \ref{thm:Pareto} implies that in our environment with
optimizing agents, it is actually welfare-improving for individuals
to ``overweight'' their own information relative to best-response
behavior.

The condition on equilibrium weights says that no agent anti-imitates
any of her neighbors. This assumption makes the analysis tractable,
but we believe the basic force also works in finite networks with
some anti-imitation. In the proof in Appendix \ref{OA-sec:remaining-proofs}, we state and
prove a more general result where weights are non-negative but need not all arise from Bayesian or naive updating.

\medskip{}

\paragraph{Proof sketch}

The idea of the proof for the Bayesian equilibrium case is to begin at the steady
state and then marginally shift each agent's weights toward
her private signal. This means agents' actions
are less correlated but, by the envelope theorem, not significantly worse in the next period.
We show that if all agents continue using these new weights, the decreased
correlation eventually benefits everyone. To do this, we use
the absence of anti-imitation, which implies a certain monotonicity in the updating function whereby the initial decrease in correlation results in all agents' variances decreasing. 

The proof in the naive case is simpler. Here a naive agent is overconfident
about the quality of her social information, so she would benefit
from shifting some weight from her social information to her signal.
This deviation also reduces her correlation with other agents, so
it is Pareto-improving.

\section{Social influence}\label{sec:socialinfluence}

A canonical question about learning in networks is how much influence various agents have in affecting
aggregate behavior. This is a focus of studies including \citet*{DeMarzo2003Persuasion}
and \citet{Golub2010Naive} in the DeGroot model with an unchanging state. In this section, we
define a suitable analogue of social influence for our dynamic environment.
We then study how an agent's influence depends on her signal precision
and degree. We find that, relative to benchmark results from the DeGroot
model, influence is more sensitive to signal precisions, while social connectedness
plays a similar role in both models.

\subsection{Defining social influence}

We define the \emph{total influence} of node $i$ in a stationary equilibrium with weights $(\widehat{\bm{W}},\widehat{\bm{w}}^s)$ to be the total weight
that all actions place on the private signal of agent $(i,t)$. The total influence measures the total increase in actions
if $s_{i,t}$ increases by $1$ (due to an idiosyncratic
shock).\footnote{Note that in a stationary equilibrium, this depends on the node and not the time, so we speak interchangeably of the influence of a node and that of an agent at this node.} At equilibrium, the total influence of $i$ is:
\[
\TI(i)=\sum_{j\in N}\sum_{k=0}^{\infty}\left(\rho^k \bm{\widehat{W}}^k\right)_{ji}\widehat{w}_{i}^{s}.
\]
This expression for total influence is a version of Katz-Bonacich centrality
with respect to the matrix $\bm{\widehat{W}}$ of weights. The decay parameter
is the persistence $\rho$ of the of the AR(1) state process.

We define the \emph{social influence} of $i$ to be the total weight
that all actions \emph{in future periods} place on the private signal of agent $(i,t)$. At equilibrium, the social influence of $i$
is:
\[
\SInf(i)=\sum_{j\in N}\sum_{k=1}^{\infty}\left(\rho^k \bm{\widehat{W}}^k\right)_{ji}\widehat{w}_{i}^{s}=\TI(i)-\widehat{w}_{i}^{s}.
\]
The social influence measures the influence of an agent at node $i$ on other
agents. Social influence and total influence differ only by the weight $(i,t)$ places on her own current signal, because an agent's signal realization does not affect others' actions in the same period. Note that agent $i$'s social influence depends on the weight $\widehat{w}_{i}^{s}$
she places on her own signal as well as the weights agents place on
each others' actions.

The next result, which follows from Proposition \ref{prop:Existence}
on equilibrium existence, shows that the summation that defines social influence is guaranteed to
converge at equilibrium, which makes social influence (and similarly
total influence) well-defined.\footnote{Since $\widehat{\bm{W}}$ can contain both positive and negative numbers,
some of them potentially  large, it is not immediately obvious
that the summation converges.}
\begin{prop}
\label{cor:socialinfluence}The social influence $\SInf(i)$ is well-defined
at any equilibrium and is equal to $\left[\bm{1}^{\top} (\bm{I}-\rho\widehat{\bm{W}})^{-1}-\bm{1}^{\top} \right]_{i}\widehat{w}_{i}^{s}.$
\end{prop}

We show this as follows: if social influence did not converge, some
agents would have actions with very large variances (because their actions would depend sensitively on small idiosyncratic shocks). But then these
agents would have simple deviations that would improve their accuracy,
such as following their private signals. So this could not happen
in equilibrium. Once the infinite series defining social influence is shown to be convergent, the proposition follows by a standard Neumann series identity.

In general, social influence can be negative: an agent's net effect on others can be in the opposite direction of her signal. 

\subsection{Which agents are influential?}

We now ask how the social influence $\SInf(i)$ of an agent depends on
her signal precision and degree. To facilitate the most direct comparison
with standard results in models with a fixed state, such as \citet*{DeMarzo2003Persuasion}, we focus on cases where social influences are positive.

To examine the effect of signal precision on social influence, we first study
complete networks with $n\geq2$ agents and two private signal variances:
half the agents have more precise signals, and the other half have
less precise signals. We call the two groups' signal variances $\sigma_{A}^{2}$
and $\sigma_{B}^{2}$ and the corresponding agents' social influences
$\SInf(A)$ and $\SInf(B)$. We show that the ratio between the two groups'
social influences in equilibrium is larger than the ratio between their signal precisions
(whenever the imprecise group has positive social influence).
\begin{prop}
\label{prop:SIComplete}On a complete network with $m=1$ and signal
variances $\sigma_{A}^{2}<\sigma_{B}^{2}$, in the unique equilibrium it holds that
\[
\frac{\SInf(A)}{\SInf(B)}>\frac{\sigma_{A}^{-2}}{\sigma_{B}^{-2}}
\]
 whenever $\SInf(B)>0$.
\end{prop}
The proposition says that increasing a group's precision increases their influence more than proportionately. As we have seen in our main results, if the precision difference is large enough, then it is optimal to place zero or negative weight on the less precise group. The result says that even before this happens, imprecision reduces a group's influence considerably---and, as we will discuss below, more than in benchmark models of social influence. 

The proposition assumes the network is complete, but numerical evidence
suggests that on other networks, too, agents with more precise signals tend to be much more
influential. We simulate a configuration model with
$n=40$ nodes, each with degree $d=5$.\footnote{This model works by creating $n$ nodes, each with $d$ ``stubs'' sticking out of it, and then performing a random matching of the stubs to create a graph. See \citet{jackson2010social}, Section 4.5.10, for details.}
Nodes are randomly assigned
to have a precise signal with variance $\sigma_{A}^{2}$ or an imprecise
signal with variance $\sigma_{B}^{2}$ (with equal probability).

We are interested in the ratio ${\SInf(A)}/{\SInf(B)}$ in this more complicated environment. If social influence were approximately proportional to precision, then ${\SInf(A)}/{\SInf(B)}$ would be approximately ${\sigma_{A}^{-2}}/{\sigma_{B}^{-2}}$. To assess by how much the influence of the precise group exceeds the level suggested by this benchmark, we will look at the ratio 
$$R_{\sigma}=\frac{{\SInf(A)}/{\SInf(B)}}{{\sigma_{A}^{-2}}/{\sigma_{B}^{-2}}}.$$
Table
\ref{sigvartable} reports this ratio over $100$ runs of the simulation model for various pairs of
$\sigma_{A}^{2}$ and $\sigma_{B}^{2}$, each in the interval $[0.5,5]$. The entries of the table would be equal to $1$ if influence is proportional to precision. Instead, all off-diagonal entries are greater than one (or negative), meaning social influence depends more (and often
much more) on signal precision than in the proportional benchmark.

\begin{table}

{\tiny
\begin{tabular}{@{}l|l| *{10} {S[table-format=5.1]}@{}}
\toprule
 \multicolumn{10}{c}{\hspace{5cm}$\sigma_B^2$} \\
 \cline{1-12}

 & Precision & 0.5 & 1 & 1.5 & 2 & 2.5 & 3 & 3.5 & 4 & 4.5 & 5 \\ 
\cline{2-12}

 \multirow{10}{*}{$\sigma_A^2$} 
 
&0.5 & 1 & 2.14 & 4.53 & 9.06 & 168.09 & -29.69 & -15.54 & -8.22 & -7.41 & -6.19 \\ 
&1 & & 1 & 1.80 & 3.15 & 5.62 & 9.95 & 26.35 & -29.89 & -15.23 & -13.67 \\ 
&1.5 & & & 1 & 1.64 & 2.56 & 3.80 & 6.27 & 10.64 & 44.97 & -31.32 \\ 
&2 & & & & 1 & 1.51 & 2.22 & 3.13 & 4.35 & 7.05 & 11.34 \\ 
&2.5 & & & & & 1 & 1.43 & 2.04 & 2.69 & 3.88 & 5.50 \\ 
&3 &  & & & & & 1 & 1.38 & 1.86 & 2.46 & 3.35 \\ 
&3.5 & & & & & & & 1 & 1.34 & 1.74 & 2.27 \\ 
&4 & &  & & & & & & 1 & 1.30 & 1.67 \\ 
&4.5  & & & & & & & & & 1 & 1.28 \\ 
&5 & &  & & & & & & & & 1 \\ 
\bottomrule
\end{tabular}}

\caption{\footnotesize The table shows how far influence ratios are from a benchmark of being proportional to precision.  We use a configuration model with a regular network and heterogeneous
signal variances; there are $n=40$ agents and the degree is $d=5$.
Agents are randomly assigned to signal variances $\sigma_{A}^{2}$
or $\sigma_{B}^{2}$. Each entry is computed from $100$ runs with persistence $\rho=0.9$. Each table entry reports the ratio $R_{\sigma}=\frac{{\SInf(A)}/{\SInf(B)}}{{\sigma_{A}^{-2}}/{\sigma_{B}^{-2}}}$ for the precision parameters corresponding to that entry.}
\label{sigvartable}
\end{table}

Having examined how influence depends on precisions, we turn to how it depends on degrees. We again use a configuration model, which allows us to fix any desired empirical degree distribution and generate the graphs uniformly conditional on the degrees.
We will find that social influence
depends less on degree than on precision. We simulate a configuration model with $n=40$
nodes, each randomly assigned degrees $d_{A}$ or $d_{B}$ (with equal
probability of each) and with $\sigma^{2}=2$ for all agents. Table \ref{degtable}
reports the ratio $$R_{d}=\frac{{\SInf(A)}/{\SInf(B)}}{{d_A/d_B}}.$$
over $100$ runs of the simulation model for degrees between $1$
and $10$. Again, the entries would be equal to $1$ if social influence were proportional to degree. Social influence is indeed approximately proportional to degree: the entries in the table range between $0.91$ and $1.11$.

\begin{table}

{\tiny 
\begin{tabular}{@{}l|l| *{10} {S[table-format=5.1]}@{}}
\toprule
 \multicolumn{10}{c}{\hspace{5cm}$d_B$} \\ 
 \cline{1-12}

& Degree &1 & 2 & 3 & 4 & 5 & 6 & 7 & 8 & 9 & 10 \\  
\cline{2-12} 
\multirow{10}{*}{$d_A$}

& 1 & 1&  1.11 &  1.02 & 0.96 & 0.94 & 0.92 & 0.96 & 1.00 &  1.03 &  1.09 \\ 
& 2 &  &   1&  1.02 &  1.01 & 0.98 & 0.96 & 0.92 &  0.93 & 0.91 & 0.95 \\ 
& 3 & & &   1&  1.01 &  1.01 &  1.01 & 0.98 & 0.96 & 0.94 & 0.93 \\ 
& 4 & & & &   1&  1.01&  1.01 &  1.00 &  0.99 & 0.98 & 0.96 \\ 
& 5 & & & & &   1&  1.01 &  1.01 &  1.02 &  1.01 &  1.00 \\ 
& 6 & & & & & &   1&  1.01 &  1.01 &  1.01 &  1.01 \\ 
& 7 & & & & & & &   1&  1.01&  1.01&  1.02 \\ 
& 8 & & & & & & & &   1&  1.01&  1.01 \\ 
& 9 & & & & & & & & &   1&  1.01 \\ 
& 10 && & & & & & & & &   1 \\  
\bottomrule  \end{tabular}}

\caption{\footnotesize The table shows how far influence ratios are from a benchmark of being proportional to degree.  We use a configuration model with two possible degrees on $n=40$ agents with homogeneous signal variance $\sigma^2=2$.
Agents are randomly assigned to degrees $d_A$
or $d_B$. Each entry is computed from $100$ runs with persistence $\rho=0.9$.  Each table entry reports the ratio $R_{d}=\frac{{\SInf(A)}/{\SInf(B)}}{{d_A/d_B}}$  for the degree parameters corresponding to that entry. }
\label{degtable}
\end{table}

\begin{rem} A simple intuition explains why social influence depends more on private information than on network position. Increasing an agent's private signal precision and her degree both tend to make her action more accurate. Increasing private signal precision has the additional effect of increasing an agent's weight on her private information, which is recent and independent of other agents' actions. This provides more reason for others to place weight on her actions, amplifying the effect of the increased accuracy. In contrast, increasing degree tends to make an agent place more weight on her social information, which is older and more correlated with others. This countervails the effect of increased accuracy, making the agent a less appealing source for others. \label{rem:reason} \end{rem}

The exercises so far varied only one of signal precision or degree, and we now explore how social influence depends on precision and degree jointly. To do so, we compute equilibrium social influences on 5,000 networks with $n=40$ agents in each. Each agent is randomly assigned a degree chosen uniformly from $\{1,2,\hdots,7\}$ and a private signal variance chosen uniformly and independently from $\{0.5,1,\hdots,3.5\}$. Networks are then drawn via the configuration model. Figure \ref{fig:levelcurves} plots the level curves for average social influence conditional on node attributes. The steepness of the level curves shows that social influence again depends more on signal variance than degree, especially when signals are less precise.

\begin{figure}
\begin{centering}
\includegraphics[scale=0.5]{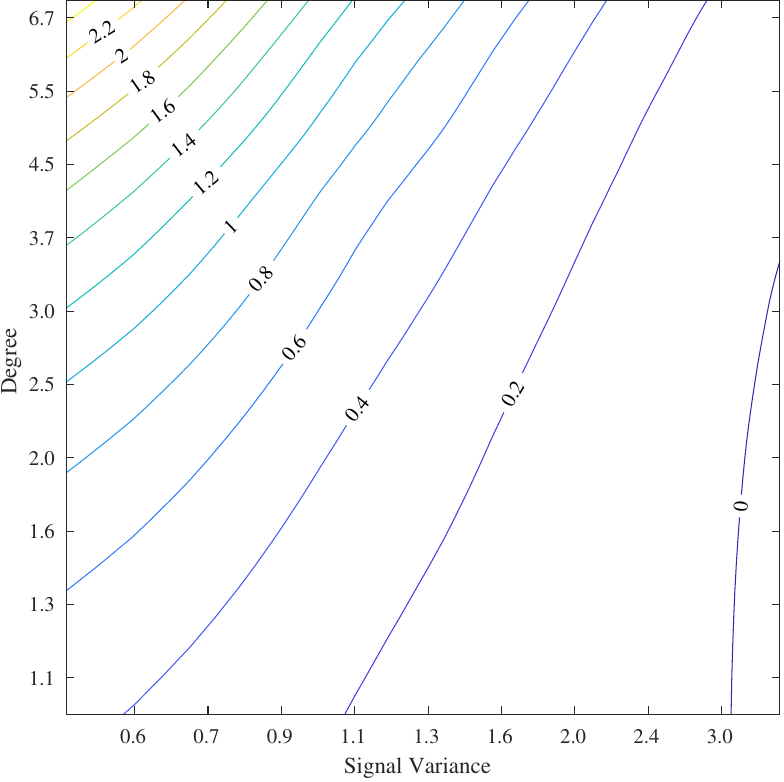}
\begin{centering}
\caption{\label{fig:levelcurves} \footnotesize Level curves for average social influence of agents in a configuration model with 5,000 networks with $n=40$ agents in each and persistence $\rho=0.9$. Degrees are chosen uniformly from $\{1,2,\hdots,7\}$ and private signal variances are chosen uniformly from $\{0.5,1,\hdots,3.5\}$. The figure shows level curves for average social influence (drawn via cubic interpolation) on a log-log plot. If proportional changes in degree and signal variances mattered equally, these level curves would have slope $1$.}
\par\end{centering}

\par\end{centering}

\end{figure}

\subsection{Comparison with a DeGroot benchmark}

The results above are interesting to compare with those of canonical network models with a fixed state. A relevant benchmark is a version of the DeGroot model studied by \citet*{DeMarzo2003Persuasion}.
Agents start with an improper prior, receive independent normal private
signals $s_{i}$ (with different precisions) about the state once, and then each takes an action
$a_{i,0}$ equal to her expectation of the state $\theta$. After this, agents
observe their neighbors' actions and take actions $a_{i,1}$, which
are Bayesian expectations of the state $\theta$ given their observations.
In all subsequent periods $t>1$, agents observe their neighbors'
actions $a_{j,t-1}$ and take actions $a_{i,t}$ as if $a_{j,t-1}$
had the same distribution as $j$'s private signal. That is,  they naively repeat their optimal strategy from the first period, which \citet{DeMarzo2003Persuasion} interpret as a quasi-Bayesian, boundedly rational procedure.

One natural measure of social influence is the influence of $s_{i}$
on the long-run consensus estimate $\lim_{t\to \infty} a_{j,t}$ held by any agent. In an undirected,
connected, and aperiodic network, this limit exists and the influence
of agent $i$ is proportional to her private signal precision $\sigma_{i}^{-2}$
and to her degree $d_{i}$. Compared to this benchmark, social influence
in our changing-state model is more sensitive to signal precision
(in the complete graph and in our simulations for configuration models).
On the other hand, the dependence of social influence on degree is very similar to the DeGroot benchmark---approximately proportional. To summarize, influence depends more on an agent's private information, while the dependence on network position is remarkably similar. The difference between the benchmark and our model  is explained in Remark \ref{rem:reason}.

\section{Related literature\label{sec:Related-literature}\label{sec:DeGroot}}

Whether decentralized communication can facilitate
efficient adaptation to a changing world is a fundamental question in economic
theory, related to questions raised by \citet{hayek1945use}\footnote{``If \ldots the economic problem of society is mainly
one of rapid adaptation to changes in the particular circumstances
of time and place \ldots there still remains the problem of communicating
to {[}each individual{]} such further information as he needs.''
Hayek's main concern was aggregation of information through markets,
but the same questions apply more generally.} and central to certain applied problems, e.g., in real business
cycle models with consumers and firms learning about evolving states.\footnote{See \citet{angeletos2010noisy} for a survey of related models that are used to study real business cycles. More recent developments include \citet{angeletos2018forward} and 
\cite{molavi2019macroeconomics}, with the latter allowing a form of misspecification.
}
Nevertheless, there is relatively little modeling of Bayesian learning
of dynamic states in the large literature on social learning and information
aggregation in networks, whose most relevant papers we now review.\footnote{For more complete surveys of different parts of this literature, see,
among others, \citet{acemoglu2011opinion}, \citet{golub-sadler},
and \citet{mossel2017opinion}. See \citet*{Moscarini1998Social}
for an early model in a binary-action environment, where it is shown
that a changing state can break information cascades.}

Play in the stationary linear equilibria of our model closely resembles
behavior in the \citet{degroot1974reaching} model, where agents update
by linearly aggregating network neighbors' past estimates, with constant
weights on neighbors over time. \citet*{DeMarzo2003Persuasion}, in
a Gaussian environment with an unchanging state, derive DeGroot learning as the Bayesian behavior in the first round of communication, and use that as a foundation for a DeGroot rule as a boundedly-rational heuristic. 
\citet*{molavi2018theory} present new bounded-rationality foundations for the DeGroot rule. Our different environment offers a different foundation for averaging rules with time-invariant weights: as a stationary equilibrium of a stationary environment.\footnote{Indeed, agents behaving according to the DeGroot heuristic in other environments might have to do with their experiences in stationary
environments where it is closer to optimal.} Though the updating rule resembles those studied in fixed-state environments, we have stressed that the learning implications are quite different. 

Several recent papers in engineering and computer science study dynamic
environments similar to ours. \citet*{Shahrampour13Online} study
an exogenous-weights version, interpreted as a set of Kalman filters
under the control of a planner. They bound
measures of welfare in terms of  the
persistence of the state process ($\rho$) and network properties, such as the spectral gap. \citet*{Frongillo2011Social}
study a $\rho=1$ model of the state. They characterize the steady-state
distribution of behavior for any weights, and calculate equilibrium weights on a complete network, which they show are inefficient. Our Proposition \ref{thm:Pareto} documents a related inefficiency;  the quality of equilibrium learning in large, incomplete networks and social influence in equilibrium are  topics not considered in these papers. In economics, \citet*{alatas2016network} perform an empirical exercise in a similar model with a quasi-Bayesian learning rule. Their estimation assumes agents ignore the correlations between social observations, similarly to our naive models.\footnote{The paper's focus is estimating parameters of social learning rules
using data from Indonesian villages, where agents are trying to estimate each other's wealth.} Our results show that the degree of rationality can be
pivotal for the outcomes of such processes, and provide foundations for structural inference to test various behavioral assumptions.

Our results about when agents learn well are related to two phenomena that have played an important role in the social learning literature. One theme in this literature is that heterogeneity---in agents' neighborhoods or preferences---can be helpful for learning. A manifestation of this is the usefulness  of \emph{sacrificial lambs} (typically studied in sequential social learning models with a fixed state): a small set of agents who observe nobody can help everyone else learn well, because their actions are then informative only about their private signals, and unconfounded by an information cascade
(\citealp*{sgroi2002optimizing}, \citealp*{arieli2019multidimensional}). Heterogeneity in preferences can serve a similar purpose: if preferences have full support, there is a positive probability that preference bias counteracts available social information, causing an agent to follow her private signal (\citealp*{goeree2006social}, \citealp*{lobel2016preferences}). A crucial difference is that our mechanism does not rely on any agents simply revealing their private signals:
heterogeneity helps by changing how neighbors use their \emph{social} information, which in turn aids an agent in inferring a common confound.\footnote{A bit farther afield, in \citet{sethi2012public}, learning outcomes when two individuals repeatedly learn from each other depend on whether their
(heterogeneous) priors are independent or correlated; the common thread is that a natural assumption about agents' attributes (independent priors in their case) leads to an identification problem. The mechanics are otherwise quite different.}

Second, a robust aspect of rational learning in sequential models
is the phenomenon of anti-imitation, as discussed, e.g., by \citet{eyster2014extensive}.
They give general conditions for fully Bayesian agents to anti-imitate
in a standard sequential model. We find that anti-imitation is also an important
feature in our dynamic model, and in our context is crucial for good
learning. Despite this similarity, there is an important contrast
between our environment and standard sequential models. In those models,
while rational agents \emph{do} prefer to anti-imitate, 
individuals and society as a whole can often obtain good outcomes using
heuristics without any anti-imitation: for instance, by combining one's
own private signal with
the information that can be inferred from a single neighbor. \citet*{AcemogluNetwork} and \citet{lobel2015information}
show that such a heuristic leads to asymptotic learning in a sequential
model. Our dynamic learning environment is different, as shown in
Proposition \ref{prop:asymundir}: to have any hope of approaching
good aggregation benchmarks, agents must respond in a sophisticated
way, with anti-imitation, to their neighbors' (correlated) estimates.

\section{Discussion and extensions\label{sec:Discussion-and-extensions}}

\subsection{Aggregation and its absence without asymptotics: Numerical results} \label{sec:village_numerics_body}

The message of Section \ref{sec:LearningOutcomes} is that signal
diversity enables good aggregation, and signal homogeneity obstructs
it. The theoretical results were asymptotic, and relied on various assumptions about network structure. It is natural to ask whether our main conclusions hold up in realistic finite networks. 
To analyze this, we numerically
study equilibria of our model on graphs reflecting social relationships measured in
Indian villages \citep*{DVN/U3BIHX_2013}.
This subsection briefly summarizes our findings; we describe the
exercise fully in Appendix \ref{OA-subsec:Villages}. %

We examine the benefits of signal heterogeneity for
equilibrium aggregation. The network data are
essentially the only empirical input to our exercise.\footnote{In particular, we have no data on signal qualities; when we introduce signal heterogeneity, we simply posit that households without electricity have worse access to external information.} Given a network, we compute equilibria using our model and parameters chosen for illustration. We compare two environments that differ in signal allocations: (i) a homogeneous
case, with all signal variances set to 2, and (ii)~a heterogeneous
case, where half of the nodes have a signal variance greater than 2 (which we vary) and
half of the nodes have a signal variance less than 2.\footnote{We choose the larger signal variance so that the average precision
in each village is $\frac{1}{2}$, which holds the total inflow of information constant in a sense made precise in the appendix.}

We first compare the value of a good network with the value of heterogeneous signals. Some networks have better learning than others even with homogeneous signals.  We define the \emph{network-driven
variation} in learning to be the standard deviation of learning quality (aggregation error) across villages in the homogeneous case. Our main finding
is that increasing the private signal variance for half of the agents
by 50\%, and reducing the signal variance of the others to keep total information constant, changes social signal error variance by 6.5 times the network-driven variation. In fact, introducing this amount of private signal heterogeneity improves learning much more than the most favorable network among the villages.

Though the asymptotic prediction changes starkly depending on whether signal precisions are identical or not, considerable heterogeneity is actually required to achieve the benefits of signal diversity in a finite network. Starting from homogeneous signals and increasing signal diversity, aggregation error changes very slightly at first. Once the variance of the less precise signal has increased by 50\% relative to the starting point, learning quality has moved about halfway to what is achievable with the most extreme signal heterogeneity.

\subsection{Multidimensional states and informational specialization\label{subsec:multidim}}

Our formal analysis assumed a one-dimensional state and one-dimensional
signals, which varied only in their precisions. Our message about
the value of diversity is, however, better interpreted in a mathematically
equivalent multidimensional model.

Consider Bayesian agents who learn and communicate about two independent
dimensions simultaneously, each one working as in our model. If all
agents have equally precise signals about both dimensions, then society
may not learn well about either of them. In contrast, if half the
agents have superior signals about one dimension and inferior signals
about the other (and the other half has the reverse), then society
can learn well about both dimensions. Thus, the designer has a strong
preference for an organization with informational specialization where
some, but not all, agents are expert in a particular dimension.\footnote{This raises important questions about what information agents would
acquire, and whom they would choose to observe, which are the focus
of a growing literature. For recent papers on this in the context
of networks, see \citet{sethi2016communication} and \citet{myatt},
among others.}

Of course, there are many familiar reasons for specialization in 
having precise information about an issue. For instance, it may be that specialization is technologically efficient, or makes it easier to provide incentives. Crucially, specialization is valuable in our
setting for a distinct reason: it helps agents with
their inference problems.

More generally, one could readily extend our model and equilibrium concept to a multi-dimensional state $\theta_t \in \mathbb{R}^d$ and arbitrary Gaussian signals about it, with flexible correlations. We would expect to find suitable generalizations of the basic message that sufficient diversity within neighborhoods (in terms of signal types) facilitates learning. The assumption that agents know neighbors' signal distributions is clearly very helpful for tractability; it would be interesting to consider models in which agents are also uncertain about these distributions.
\\ \\
\textbf{Supplementary Data}

The data and code underlying this article are available on Zenodo, at \url{https://dx.doi.org/10.5281/zenodo.6954517}.

{ {\footnotesize{}\bibliographystyle{ecta}
\bibliography{MS28075bibliogaphy}
 }}{\footnotesize\par}

\appendix

\section{Existence of equilibrium: Proof of Proposition \ref{prop:Existence}\label{sec:Existence-of-equilibrium:appendix}}

Recall from Section \ref{subsec:Best-Response-Weights} the map $\Phi$,
which gives the next-period covariance matrix $\Phi(\bm{V}_{t})$
for any $\bm{V}_{t}$. The expression given there for this map ensures
that its entries are continuous functions of the entries of $\bm{V}_{t}$.
Our strategy is to show that this function maps a convex, compact
set, $\mathcal{K},$ to itself, which, by Brouwer's fixed-point theorem,
ensures that $\Phi$ has a fixed point $\bm{\widehat{V}}$. We will
then argue that this fixed point corresponds to a stationary linear
equilibrium.

We begin by defining the compact set $\mathcal{K}$. Recalling Section \ref{sec:BRweights}, entries of $\bm{V}_{t}$ are covariances between pairs
of neighbor errors from any periods $t-\ell$ where $1 \leq \ell \leq m$. Let $k,l$
be two indices of such actions, corresponding to actions taken at
nodes $i$ and $j$ respectively (at potentially different times),
and let $\overline{\sigma}_{i}^{2}=\max\left\{\sigma_{i}^{2},\rho^{m-1}\sigma_{i}^{2}+\frac{1-\rho^{m-1}}{1-\rho}\right\}$.
Now let $\mathcal{K}\subset\mathcal{V}$ be the subset of symmetric
positive semi-definite matrices $\bm{V}_{t}$ such that, for any such
$k,l$,
\begin{align*}
V_{kk,t} & \in\left[\min\left\{ \frac{1}{1+\sigma_{i}^{-2}},\frac{\rho^{m-1}}{1+\sigma_{i}^{-2}}+\frac{1-\rho^{m-1}}{1-\rho}\right\} ,\max\left\{ \sigma_{i}^{2},\rho^{m-1}\sigma_{i}^{2}+\frac{1-\rho^{m-1}}{1-\rho}\right\} \right]\\
V_{kl,t} & \in\left[-\overline{\sigma}_{i}\overline{\sigma}_{j},\overline{\sigma}_{i}\overline{\sigma}_{j}\right].
\end{align*}
This set is closed and convex, and we claim that $\Phi(\mathcal{K})\subset\mathcal{K}.$

To show this claim, we will first find upper and lower bounds on the
variance of any neighbor's action (at any period in memory). For the
upper bound, note that a Bayesian agent will not choose an action
with a larger variance than her signal, which has variance $\sigma_{i}^{2}$.
For a lower bound, note that if she knew the previous period's state
and her own signal, then the variance of her action would be $\frac{1}{1+\sigma_{i}^{-2}}$.
Thus an agent observing only noisy estimates of $\theta_{t}$ and
her own signal can do no better.

By the same reasoning applied to the node-$i$ agent from $m$ periods
ago, the error variance of $\rho^{m}a_{i,t-m}-\theta_{t}$ is at most
$\rho^{m}\sigma_{i}^{2}+\frac{1-\rho^{m}}{1-\rho}$ and at least $\frac{\rho^{m}}{1+\sigma_{i}^{-2}}+\frac{1-\rho^{m}}{1-\rho}$.
This establishes bounds on $V_{kk,t}$ for observations $k$ from
either the most recent or the oldest available period. The corresponding
bounds from the periods between $t-m+1$ and $t$ are always weaker
than at least one of the two bounds we have described, so we need
only take minima and maxima over two terms.

This established the claimed bound on the variances. The bounds on
covariances follow from Cauchy-Schwartz.

\medskip{}

We have now established that there is a variance-covariance matrix
$\bm{\widehat{V}}$ such that $\Phi(\bm{\widehat{V}})=\bm{\widehat{V}}$.
By definition of $\Phi$, this means there exists some weight profile
$(\text{\ensuremath{\bm{\widehat{W}}}},\bm{\widehat{w}^{s}})$ such
that, when applied to prior actions that have variance-covariance
matrix $\bm{\widehat{V}}$, produce variance-covariance matrix $\bm{\widehat{V}}$.
However, it still remains to show that this is the variance-covariance
matrix reached when agents have been using the weights $(\text{\ensuremath{\bm{\widehat{W}}}},\bm{\widehat{w}^{s}})$
forever.

To show this, first observe that if agents have been using the weights
$(\text{\ensuremath{\bm{\widehat{W}}}},\bm{\widehat{w}^{s}})$ forever,
the variance-covariance matrix $\bm{V}_{t}$ in any period is uniquely
determined and does not depend on $t$; call this $\check{\bm{V}}$.\footnote{ The variance-covariance matrices are well-defined because the $(W,w^{s})$
weights yield unambiguous strategy profiles in the sense of Appendix
\ref{OA-sec:Details-of-definitions}.} This is because actions can be expressed as linear combinations of
private signals with coefficients depending only on the weights. Second,
it follows from our construction above of the matrix $\bm{\widehat{V}}$
and the weights $(\text{\ensuremath{\bm{\widehat{W}}}},\bm{\widehat{w}^{s}})$
that there is a distribution of actions where the variance-covariance
matrix is $\widehat{\bm{V}}$ in every period and agents are using
weights $(\text{\ensuremath{\bm{\widehat{W}}}},\bm{\widehat{w}^{s}})$
in every period. Combining the two statements shows that in fact $\check{\bm{V}}=\widehat{\boldsymbol{V}}$,
and this completes the proof. Note that this argument also establishes
that the response profile we have constructed is a strategy profile:
under the responses used, we can write formally the dependence of
actions on all prior signals, and verify using the observations on
decay of dependence across time that the formula is summable and hence
defines unique actions.

\section{Proof of Theorem \ref{thm:DiverseSignals} \label{sec:Proof-of-Theorem-diverse}}

\subsection{Notation and key notions\label{subsec:Notation-and-key}}

Let $\mathbb{S}$ be the (by assumption finite) set of all possible
signal variances, and let $\overline{\sigma}^{2}$ be the largest
of them. The proof will focus on the covariances of errors in social
signals. Suppose that all agents have at least one neighbor. Take
two arbitrary agents $i$ and $j$. Recall that both $r_{i,t}$ and
$r_{j,t}$ have mean $\theta_{t-1}$, because each is an unbiased
estimate\footnote{This is because it is a linear combination, with coefficients summing
to $1$, of unbiased estimates of $\theta_{t-1}$.} of $\theta_{t-1}$; we will thus focus on the errors $r_{i,t}-\theta_{t-1}$.
Let $\bm{A}_{t}$ denote the variance-covariance matrix $\left(\Covar(r_{i,t}-\theta_{t-1},r_{j,t}-\theta_{t-1})\right)_{i,j}$
and let $\mathcal{W}$ be the set of such covariance matrices. For
all $i,j$ note that $\Covar(r_{i,t}-\theta_{t-1},r_{j,t}-\theta_{t-1})\in[-\overline{\sigma}^{2},\overline{\sigma}^{2}]$
using the Cauchy-Schwarz inequality and the fact that $\text{Var}(r_{i,t}-\theta_{t-1})\in[0,\overline{\sigma}^{2}]$
for all $i$. This fact about variances says that no social signal
is worse than putting all weight on an agent who follows only her
private signal. Thus the best-response map $\Phi$ is well-defined
and induces a map $\widetilde{\Phi}$ on $\mathcal{W}$.

Next, for any $\psi,\zeta>0$ we will define the subset $\mathcal{W}_{\psi,\zeta}\subset\mathcal{\mathcal{W}}$
to be the set of covariance matrices in $\mathcal{\mathcal{W}}$ such
that both of the following hold:
\begin{enumerate}
\item for any pair of distinct agents\footnote{ Throughout this proof, we abuse terminology by referring to agents
and nodes interchangeably when the relevant $t$ is clear or specified
nearby.} $i\in G_{n}^{k}$ and $j\in G_{n}^{k'}$, 
\[
\Covar(r_{i,t}-\theta_{t-1},r_{j,t}-\theta_{t-1})=\psi_{kk'}+\zeta_{ij}
\]
where (i) $\psi_{kk'}$ depends only on the network types of the two
agents ($k$ and $k'$, which may be the same); (ii) $|\psi_{kk'}|<\psi$;
and (iii) $|\zeta_{ij}|<\zeta$;
\item for any single agent $i\in G_{n}^{k}$, 
\[
\text{Var}(r_{i,t}-\theta_{t-1})=\psi_{k}+\zeta_{ii}
\]
where (i) $\psi_{k}$ only depends on the network type of the agent;
(ii) $|\psi_{k}|<\psi,$ and (iii) $|\zeta_{ii}|<\zeta$. 
\end{enumerate}
This is the space of covariance matrices such that each covariance
is split into two parts. Considering (1) first, $\psi_{kk'}$ is an
effect that depends only on $i$'s and $j$'s network types, while
$\zeta_{ij}$ adjusts for the individual-level heterogeneity arising
from different link realizations. The description of the decomposition
in (2) is analogous.

\subsection{Proof strategy}

\subsubsection{A set $\mathcal{W}_{\overline{\psi},\overline{\zeta}}$ of outcomes
with good learning}

Our goal is to show that as $n$ grows large, there is an equilibrium
in which $\text{Var}(r_{i,t}-\theta_{t-1})$ becomes very small, which
then implies that the agents asymptotically learn. To this end we
define a set of covariances with this property as well as some other
useful properties. We will take $\overline{\psi}$ and $\overline{\zeta}$
to be arbitrarily small numbers and show that for large enough $n$,
with high probability (which we abbreviate ``asymptotically almost
surely'' or ``a.a.s.'') there is an equilibrium with a social error
covariance matrix $\bm{A}_{t}$ in the set $\mathcal{W}_{\overline{\psi},\overline{\zeta}}$.
That will imply that, in this equilibrium, $\text{Var}(r_{i,t}-\theta_{t-1})$
becomes arbitrarily small as we take the constants $\overline{\psi}$
and $\overline{\zeta}$ to be small. In our constructions, the $\zeta_{ij}$
(resp., $\zeta_{i}$) terms will be set to much smaller values than
the $\psi_{kk'}$ (resp., $\psi_{k}$) terms, because group-level
covariances are more predictable and less sensitive to idiosyncratic
realizations than individual-level covariances.

\subsubsection{Approach to showing that $\mathcal{W}_{\overline{\psi},\overline{\zeta}}$
contains an equilibrium\label{subsec:Approach-to-proof}}

To show that there is (a.a.s.) an equilibrium outcome with a social
error covariance matrix $\bm{A}_{t}$ in the set $\mathcal{W}_{\overline{\psi},\overline{\zeta}}$,
the plan is to construct a set so that (a.a.s.) $\overline{\mathcal{W}}\subset\mathcal{W}_{\overline{\psi},\overline{\zeta}}$
and $\widetilde{\Phi}(\overline{\mathcal{W}})\subset\overline{\mathcal{W}}$.
This set will contain an equilibrium by the Brouwer fixed point theorem,
and therefore so will $\mathcal{W}_{\overline{\psi},\overline{\zeta}}$.

To construct the set $\overline{\mathcal{W}}$, we will fix a positive
constant $\beta$ (to be determined later), and define 
\[
\overline{\mathcal{W}}=\mathcal{\mathcal{W}}_{\frac{\beta}{n},\frac{1}{n}}\cup\widetilde{\Phi}\left(\mathcal{\mathcal{W}}_{\frac{\beta}{n},\frac{1}{n}}\right).
\]
We will then prove that, for large enough $n$, (i) $\widetilde{\Phi}(\overline{\mathcal{W}})\subseteq\overline{\mathcal{W}}$
and (ii) for another suitable positive constant $\lambda$, 
\[
\overline{\mathcal{W}}\subset\mathcal{W}_{\frac{\beta}{n},\frac{\lambda}{n}}.
\]
This will allow us to establish that (a.a.s.) $\overline{\mathcal{W}}\subset\mathcal{W}_{\overline{\psi},\overline{\zeta}}$
and $\widetilde{\Phi}(\overline{\mathcal{W}})\subset\overline{\mathcal{W}}$,
with $\overline{\psi}$ and $\overline{\zeta}$ being arbitrarily
small numbers.

The following two lemmas will allow us to deduce (immediately after
stating them) properties (i) and (ii) of $\overline{\mathcal{W}}$.
\begin{lem}
\label{lem:2}There is a function $\underline{\lambda}(\beta)\geq1$
such that the following holds. For all large enough $\beta$ and all
$\lambda\geq\underline{\lambda}(\beta)$, for $n$ sufficiently large
we have $\widetilde{\Phi}\left(\mathcal{\mathcal{W}}_{\frac{\beta}{n},\frac{1}{n}}\right)\subset\mathcal{W}_{\frac{\beta}{n},\frac{\lambda}{n}}$
with probability at least $1-\frac{1}{n}$.
\end{lem}
\begin{lem}
\label{lem:1}For all large enough $\beta$, for $n$ sufficiently
large, $\widetilde{\Phi}^{2}\left(\mathcal{\mathcal{W}}_{\frac{\beta}{n},\frac{1}{n}}\right)\subset\mathcal{\mathcal{W}}_{\frac{\beta}{n},\frac{1}{n}}$,
with probability at least $1-\frac{1}{n}$.\footnote{The notation $\widetilde{\Phi}^{2}$ means the operator $\widetilde{\Phi}$
applied twice.}
\end{lem}
Putting these lemmas together, a.a.s. we have, 
\[
\widetilde{\Phi}^{2}\left(\mathcal{\mathcal{W}}_{\frac{\beta}{n},\frac{1}{n}}\right)\subset\mathcal{\mathcal{W}}_{\frac{\beta}{n},\frac{1}{n}}\quad\text{and}\quad\widetilde{\Phi}\left(\mathcal{\mathcal{W}}_{\frac{\beta}{n},\frac{1}{n}}\right)\subset\mathcal{W}_{\frac{\beta}{n},\frac{\lambda}{n}}.
\]
From this it follows that $\overline{\mathcal{W}}=\mathcal{\mathcal{W}}_{\frac{\beta}{n},\frac{1}{n}}\cup\widetilde{\Phi}\left(\mathcal{\mathcal{W}}_{\frac{\beta}{n},\frac{1}{n}}\right)$
is mapped to a subset of itself by $\widetilde{\Phi}$, and contained
in $\mathcal{W}_{\frac{\beta}{n},\frac{\lambda}{n}}$, as claimed.

\subsubsection{Proving the lemmas by analyzing how $\widetilde{\Phi}$ and $\widetilde{\Phi}^{2}$
act on sets $\mathcal{W}_{\psi,\zeta}$ }

The lemmas are about how $\widetilde{\Phi}$ and $\widetilde{\Phi}^{2}$
act on the covariance matrix $\bm{A}_{t}$, assuming it is in a certain
set $\mathcal{W}_{\psi,\zeta}$, to yield new covariance matrices.
Thus, we will prove these lemmas by studying two periods of updating.
The analysis will come in five steps.

\paragraph{Step 1: No-large-deviations (NLD) networks and the high-probability
	event}

Step 1 concerns the ``with high probability'' part of the lemmas.
In the entire argument, we condition on the event of a \emph{no-large-deviations
(NLD)} network realization, which says that certain realized statistics
in the network (e.g., number of paths between two nodes) are close
to their expectations. The expectations in question depend only on
agents' types. Therefore, on the NLD realization, the realized statistics
do not vary much based on which exact agents we focus on, but rather
depend only on their types. Step 1 defines the NLD event $E$ formally
and shows that it has high probability. We use the structure of the
NLD event throughout our subsequent steps, as we mention below.

\paragraph{Step 2: Weights in one step of updating are well-behaved}

We are interested in $\widetilde{\Phi}$ and $\widetilde{\Phi}^{2}$,
which describe how the covariance matrix $\bm{A}_{t}$ of social signal
errors changes under updating. How this works is determined by the
``basic'' updating map $\Phi$, and so we begin by studying the
weights involved in it and then make deductions about the implications
for the evolution of the variance-covariance matrix $\bm{A}_{t}$.

The present step establishes that in one step of updating, the weight
$W_{ij,t+1}$ that agent $(i,t+1)$ places on the action of another
agent $j$ in period $t$, does not depend too much on the identities
of $i$ and $j$. It only depends on their (network and signal) types.
This is established by using our explicit formula for weights in terms
of covariances. We rely on (i) the fact that covariances are assumed
to start out in a suitable $\mathcal{W}_{\psi,\zeta}$, and (ii) our
conditioning on the NLD event $E$. The NLD event is designed so that
the network quantities that go into determining the weights depend
only on the types of $i$ and $j$ (because the NLD event forbids
too much variation within type). The restriction to $\bm{A}_{t}\in\mathcal{W}_{\psi,\zeta}$
ensures that covariances in the initial period $t$ do not vary too
much with type, either.

\paragraph{Step 3: Lemma \ref{lem:2}: $\widetilde{\Phi}\left(\mathcal{\mathcal{W}}_{\frac{\beta}{n},\frac{1}{n}}\right)\subset\mathcal{W}_{\frac{\beta}{n},\frac{\lambda}{n}}$}

Once we have analyzed one step of updating, it is natural to consider
the implications for the covariance matrix. Because we now have a
bound on how much weights can vary after one step of updating, we
can compute bounds on covariances. We show that if covariances $\bm{A}_{t}$
are in $\mathcal{\mathcal{W}}_{\frac{\beta}{n},\frac{1}{n}}$, then
after one step, covariances are in $\mathcal{W}_{\frac{\beta}{n},\frac{\lambda}{n}}$.
Note that the introduction of another parameter $\lambda$ on the
right-hand side implies that this step might worsen our control on
covariances somewhat, but in a bounded way. 

\paragraph{Step 4: Weights in two steps of updating are well-behaved}

Here we establish that the statement made in Step 2 remains
true when we replace $t+1$ by $t+2$. By the same sort of reasoning
as in Step 2, an additional period of updating cannot create too much
further idiosyncratic variation in weights. Proving this requires
analyzing the covariance matrices of various social signals (i.e.,
the $\bm{A}_{t+1}$ that the updating induces), which is why we needed
to do Step 3 first.

\paragraph{Step 5: Lemma \ref{lem:1}: $\widetilde{\Phi}^{2}\left(\mathcal{\mathcal{W}}_{\frac{\beta}{n},\frac{1}{n}}\right)\subset\mathcal{\mathcal{W}}_{\frac{\beta}{n},\frac{1}{n}}$}

Now we use our understanding of weights from the previous steps, along
with additional structure, to show the key remaining fact. What we
have established so far about weights allows us to control the weight
that a given agent's estimate at time $t+2$ places on the social
signal of another agent at time $t$. This is Step 5(a). In the second
part, Step 5(b), we use that to control the covariances in $\bm{A}_{t+2}$.
It is important in this part of the proof that different agents have
very similar ``second-order neighborhoods'': the paths of length
2 beginning from an agent are very similar, in terms of their counts
and what types of agents they go through. We use our control of second-order
neighborhoods, as well as the assumptions on variation across entries
of $\bm{A}_{t}$ to bound this variation well enough to conclude that $\bm{A}_{t+2}\in\mathcal{\mathcal{W}}_{\frac{\beta}{n},\frac{1}{n}}$.

\subsection{Carrying out the steps}

\subsubsection{Step 1}

Here we formally define the NLD event, which we call $E$. It is given
by $E=\cap_{i=1}^{5}E_{i}$, where the events $E_{i}$ will be defined
next.

$(E_{1})$ Let $X_{i,\tau k}^{(1)}$ be the number of agents having
signal type $\tau$ and network type $k$ who are observed by $i$.
The event $E_{1}$ is that this quantity is close to its expected
value in the following sense, simultaneously for all possible values
of the subscript: 
\[
(1-\zeta^{2})\mathbb{E}[X_{i,\tau k}^{(1)}]\leq X_{i,\tau k}^{(1)}\leq(1+\zeta^{2})\mathbb{E}[X_{i,\tau k}^{(1)}].
\]

$(E_{2})$ Let $X_{ii',\tau k}^{(2)}$ be the number of agents having
signal type $\tau$ and network type $k$ who are observed by \emph{both}
$i$ and $i'$. The event $E_{2}$ is that this quantity is close
to its expected value in the following sense, simultaneously for all
possible values of the subscript:
\[
(1-\zeta^{2})\mathbb{E}[X_{ii',\tau k}^{(2)}]\leq X_{ii',\tau k}^{(2)}\leq(1+\zeta^{2})\mathbb{E}[X_{ii',\tau k}^{(2)}].
\]

$(E_{3})$ Let $X_{i,\tau k,j}^{(3)}$ be the number of agents having
signal type $\tau$ and network type $k$ who are observed by agent
$i$ and who observe agent $j$. The event $E_{3}$ is that this quantity
is close to its expected value in the following sense, simultaneously
for all possible values of the subscript:
\[
(1-\zeta^{2})\mathbb{E}[X_{i,\tau k,j}^{(3)}]\leq X_{i,\tau k,j}^{(3)}\leq(1+\zeta^{2})\mathbb{E}[X_{i,\tau k,j}^{(3)}].
\]

$(E_{4})$ Let $X_{ii',\tau k,j}^{(4)}$ be the number of agents having
signal type $\tau$ and network type $k$ who are observed by both
agent $i$ and $i'$ and who observe $j$. The event $E_{4}$ is that
this quantity is close to its expected value in the following sense,
simultaneously for all possible values of the subscript:
\[
(1-\zeta^{2})\mathbb{E}[X_{ii',\tau k',j}^{(4)}]\leq X_{ii',\tau k',j}^{(4)}\leq(1+\zeta^{2})\mathbb{E}[X_{ii',\tau k',j}^{(4)}].
\]

$(E_{5})$ Let $X_{i,\tau k,jj'}^{(5)}$ be the number of agents of
signal type $\tau$ and network type $k$ who are observed by agent
$i$ and who observe both $j$ and $j'$. The event $E_{5}$ is that
this quantity is close to its expected value in the following sense,
simultaneously for all possible values of the subscript:
\[
(1-\zeta^{2})\mathbb{E}[X_{i,\tau k,jj'}^{(5)}]\leq X_{i,\tau k,jj'}^{(5)}\leq(1+\zeta^{2})\mathbb{E}[X_{i,\tau k,jj'}^{(5)}].
\]

We claim that the probability of the complement of the event $E$
vanishes exponentially. We can check this by showing that the probability
of each of the $E_{i}$ vanishes exponentially. For $E_{1}$, for
example, the bounds will hold unless at least one agent has degree
outside the specified range. The probability of this is bounded above
by the sum of the probabilities of each individual agent having degree
outside the specified range. By Chebyshev's inequality, the probability
a given agent has degree outside this range vanishes exponentially.
Because there are $n$ agents in $G_{n}$, this sum vanishes exponentially
as well. The other cases are similar.

For the rest of the proof, we condition on the event $E$.

\subsubsection{Step 2}

As a shorthand, let $\psi=\beta/n$ for a sufficiently large constant
$\beta$, and let $\zeta=1/n$.
\begin{lem}
\label{lem:WeightVariation1}Suppose that in period $t$ the matrix
$\bm{A}=\bm{A}_{t}$ of covariances of social signals satisfies $\bm{A}\in\mathcal{W}_{\psi,\zeta}$
and all agents are optimizing in period $t+1$. Then there is a $\gamma$
so that for all $n$ sufficiently large, 
\[
\frac{W_{ij,t+1}}{W_{i'j',t+1}}\in\left[1-\frac{\gamma}{n},1+\frac{\gamma}{n}\right].
\]
whenever $i$ and $i'$ have the same network and signal types and
$j$ and $j'$ have the same network and signal types. 
\end{lem}
To prove this lemma, we will use the formula given by (\ref{eq:action_formula})  for row $i$ of the matrix $\bm{W}_{t+1}$: 
\[
W_{i\cdot,t+1}=\frac{\bm{1}^{\top}\bm{C}_{i,t}^{-1}}{\bm{1}^{\top}\bm{C}_{i,t}^{-1}\bm{1}}.
\]
This says that in period $t+1$, agent $i$'s weight on agent $j$
is proportional to the sum of the entries of column $j$ of $\bm{C}_{i,t}^{-1}$.
We want to show that the change in weights is small as the covariances
of observed social signals vary slightly. To do so we will use the
Taylor expansion of $f(\bm{A})=\bm{C}_{i,t}^{-1}$ around the covariance
matrix $\bm{A}(0)$ at which all $\psi_{kk'}=0$, $\psi_{k}=0$ and $\zeta_{ij}=0$.

We begin with the first partial derivative of $f$ at $\bm{A}(0)$ in an
arbitrary direction. Let $\bm{A}(x)$ be any perturbation of $\bm{A}(0)$ in
one parameter, i.e., $\bm{A}(x)=\bm{A}(0)+x\bm{M}$ for some constant matrix $\bm{M}$
with entries in $[-1,1]$. Let $\bm{C}_{i}(x)$ be the matrix of covariances
of the actions observed by $i$ given that the covariances of agents'
social signals were $\bm{A}(x)$. There exists a constant $\gamma_{1}$
depending only on the possible signal types such that each entry of
$\bm{C}_{i}(x)-\bm{C}_{i}(x')$ has absolute value at most $\gamma_{1}(x-x')$
whenever both $x$ and $x'$ are small.

We will now show that the column sums of $\bm{C}_{i}(x)^{-1}$ are
close to the column sums of $\bm{C}(0)_{i}^{-1}.$ To do so, we will
evaluate the formula

\begin{equation}
\frac{\partial f(\bm{A}(x))}{\partial x}=\frac{\partial\bm{C}_{i}(x)^{-1}}{\partial x}=\bm{C}_{i}(x)^{-1}\frac{\partial\bm{C}_{i}(x)}{\partial x}\bm{C}_{i}(x)^{-1}\label{eq:invder}
\end{equation}
at zero. If we can bound each column sum of this expression (evaluated
at zero) by a constant (depending only on the signal types and the
number of network types $K$), then the first derivative of $f$ will
also be bounded by a constant.

Recall that $\mathbb{S}$ is the set of signal types and let $S=|\mathbb{S}|$;
index the signal types by numbers ranging from $1$ to $S$. To bound
the column sums of $\bm{C}_{i}(0)^{-1}$, suppose that the agent observes
$r_{i}$ agents from each signal type $1\le i\leq S$. Reordering
so that all agents of each signal type are grouped together, we can
write, for a suitable matrix $\bm{Y}$ and vector $\bm{z}$:
\[
\bm{C}_{i}(0)=\begin{pmatrix}Y_{11}\bm{1}_{r_{1}\times r_{1}}+z_{1}I_{r_{1}} &  Y_{12}\bm{1}_{r_{1}\times r_{2}} &  &  Y_{S1}\bm{1}_{r_{1}\times r_{S}}\\
 Y_{12}\bm{1}_{r_{2}\times r_{1}} &  Y_{22}\bm{1}_{r_{2}\times r_{2}}+z_{2}I_{r_{2}} &  & \vdots\\
 &  & \ddots\\
 Y_{1S}\bm{1}_{r_{S}\times r_{1}} & \cdots &  &  Y_{SS}\bm{1}_{r_{S}\times r_{S}}+z_{S}I_{r_{S}}
\end{pmatrix}
\]
Therefore, the covariance matrix $\bm{C}_{i}(0)$ can be written as a block matrix with
blocks $Y_{\tau \tau'}1_{r_{\tau}\times r_{\tau'}}+z_{\tau}\delta_{\tau \tau'}I_{r_{\tau}}$ where
$1\leq \tau,\tau'\leq S$ and $\delta_{\tau \tau'}=1$ for $\tau=\tau'$ and 0 otherwise.

We now have the following important approximation of the inverse of
this matrix.\footnote{We are very grateful to Iosif Pinelis for  this argument.}
\begin{lem}[\citet{pinelis}]
\label{lem:Pinelis} Let $\bm{C}$ be a block matrix with blocks given by
\[
Y_{\tau \tau'}\bm{1}_{r_{\tau}\times r_{\tau'}}+z_{\tau}\delta_{\tau \tau'}\bm{I}_{r_{\tau}}
\] for all $\tau, \tau' \in \mathbb{S}$. As $n\rightarrow\infty$, the $(\tau,\tau)$ block of $\bm{C}^{-1}$ satisfies
\[
\frac{1}{z_{\tau}}\bm{I}_{r_{\tau}}-\frac{1}{z_{\tau}r_{\tau}}\bm{1}_{r_{\tau}\times r_{\tau}}+O(1/n^{2})
\]
while the off-diagonal blocks are $O(1/n^{2})$.
\end{lem}
\begin{proof}Note that the block $(\tau,\tau')$ of $\bm{C}^{-1}$ has the form 
\[
{E}_{\tau\tau'}\bm{1}_{r_{\tau}\times r_{\tau'}}+d_{\tau}\delta_{\tau \tau'}\bm{I}_{r_{\tau}}
\]
for some matrix ${\bm{E}}$ and vector $\bm{d}$. Here $\delta$ denotes the Kronecker delta.

Therefore, the $(\tau,\tau')$ block of $\bm{C}\bm{C}^{-1}$ can be written (using that $\bm{1}_{r\times r'} \bm{1}_{r' \times r''}= r'\bm{1}_{r \times r''}$) as 
\begin{eqnarray}
 & \sum_{ \widehat{\tau}}(Y_{\tau \widehat{\tau}}\bm{1}_{r_{\tau}\times r_{ \widehat{\tau}}}+z_{\tau}\delta_{\tau \widehat{\tau}}\bm{I}_{r_{\tau}})(   {E}_{ \widehat{\tau} \tau'}\bm{1}_{r_{ \widehat{\tau}}\times r_{ \tau'}}+d_{ \widehat{\tau}}\delta_{ \widehat{\tau} \tau'}\bm{I}_{r_{ \widehat{\tau}}})=\nonumber \\
 & \left(Y_{\tau \tau'}d_{ \tau'}+\sum_{ \widehat{\tau}}(Y_{\tau \widehat{\tau}}r_{ \widehat{\tau}}+\delta_{\tau \widehat{\tau}}z_{ \widehat{\tau}})   {E}_{ \widehat{\tau} \tau'}\right)\bm{1}_{r_{\tau}\times r_{ \tau'}}+z_{\tau}d_{\tau}\delta_{\tau \tau'}\bm{I}_{r_{\tau}}.\label{eq:CCinv-display}
\end{eqnarray}

For any vector $\bm{v} \in \mathbb{R}^{\mathbb{S}}$, let $\bm{D}_{\bm{v}}$ denote the diagonal matrix with $v_{\tau}$ in the $(\tau,\tau)$
diagonal entry and $\bm{v} \circ \bm{v}'$ denote the pointwise product of two vectors. Breaking up the fact that
(\ref{eq:CCinv-display}) equals $\bm{I}$ into its off-diagonal and diagonal parts, we have
\[
 \bm{Y}\bm{D}_{\bm{d}}+( \bm{Y}\bm{D}_{\bm{r}}+\bm{D}_{\bm{z}}){\bm{E}}= \bm{0}\text{ and }\bm{D}_{\bm{d}}=\bm{D}_{\bm{z}}^{-1}.
\]
Hence, 
\begin{eqnarray*}
{\bm{E}} & = & -( \bm{Y}\bm{D}_{\bm{r}}+\bm{D}_{\bm{z}})^{-1} \bm{Y}\bm{D}_{\bm{d}}\\
 & = & -(\bm{I}_{S}+\bm{D}_{\bm{r}}^{-1} \bm{Y}^{-1}\bm{D}_{\bm{z}})^{-1}( \bm{Y}\bm{D}_{\bm{r}})^{-1} \bm{Y}\bm{D}^{-1}_{\bm{z}}\\
 & = & -(\bm{I}_{S}+\bm{D}_{\bm{r}}^{-1} \bm{Y}^{-1}\bm{D}_{\bm{z}})^{-1}\bm{D}^{-1}_{\bm{z} \circ \bm{r}}\\
 & = & -\bm{D}^{-1}_{\bm{z} \circ \bm{r}}+O(1/n^{2}).
\end{eqnarray*}
 Therefore as $n\rightarrow\infty$
the off-diagonal blocks of $\bm{C}^{-1}$ will be $O(1/n^{2})$ while the $(\tau,\tau)$ diagonal block is 
\[
\frac{1}{z_{\tau}}\bm{I}_{r_{\tau}}-\frac{1}{z_{\tau}r_{\tau}}\bm{1}_{r_{\tau}\times r_{\tau}}+O(1/n^{2})
\]
as desired.
\end{proof}
Using Lemma \ref{lem:Pinelis} we can analyze the column sums of\footnote{ Recall we wrote $\bm{A}(x)=\bm{A}(0)+x\bm{M}$, and in (\ref{eq:invder}) we expressed
the derivative of $f$ in $x$ in terms of the matrix we exhibit here.} $\bm{C}_{i}(0)^{-1}\bm{M}\bm{C}_{i}(0)^{-1}.$
In more detail, we use the formula of the lemma to estimate both copies
of $\bm{C}_{i}(0)^{-1}$, and then expand this to write an expression
for any column sum of $\bm{C}_{i}(0)^{-1}\bm{M}\bm{C}_{i}(0)^{-1}$. It
follows straightforwardly from this calculation that all these column
sums are $O(1/n)$ whenever all entries of $\bm{M}$ are in $[-1,1]$.

We can bound the higher-order terms in the Taylor expansion by the
same technique: by differentiating equation (\ref{eq:invder}) repeatedly
in $x$, we obtain an expression for the $k^{\text{th}}$ derivative in terms
of $\bm{C}_{i}(0)^{-1}$ and $\bm{M}$: 
\[
f^{(k)}(0)=k!\bm{C}_{i}(0)^{-1}\bm{M}\bm{C}_{i}(0)^{-1}\bm{M}\bm{C}_{i}(0)^{-1}\cdot\ldots\cdot \bm{M}\bm{C}_{i}(0)^{-1},
\]
where $\bm{M}$ appears $k$ times in the product. By the same argument
as above, we can show that the column sums of $\frac{f^{(k)}(0)}{k!}$
are bounded by a constant independent of $n$. The Taylor expansion
is
\[
f(\bm{A})=\sum_{k}\frac{f^{(k)}(0)}{k!}x^{k}.
\]
Since we take $\bm{A}\in\mathcal{W}_{\psi,\zeta},$ we can assume that
$x$ is $O(1/n)$. Because the column sums of each summand are bounded
by a constant times $x^{k}$, the column sums of $f(\bm{A})$ are bounded
by a constant.

Finally, because the variation in the column sums is $O(1/n)$ and
the weights are proportional to the column sums, each weight varies
by at most a multiplicative factor of $\gamma_{1}/n$ for some $\gamma_{1}$.
We find that the first part of the lemma, which bounded the ratios
between weights $W_{ij,t+1}/W_{i'j',t+1}$, holds.

\subsubsection{Step 3}

We complete the proof of Lemma \ref{lem:2}, which states that the
covariance matrix of $r_{i,t+1}$ is in $\mathcal{W}_{\psi,\zeta'}$.
Recall that $\zeta'=\lambda/n$ for some constant $n$, so we are
showing that if the covariance matrix of the $r_{i,t}$ is in a neighborhood
$\mathcal{W}_{\psi,\zeta},$ then the covariance matrix in the next
period is in a somewhat larger neighborhood $\mathcal{W}_{\psi,\zeta'}$.
The remainder of the argument then follows by the same arguments as
in the proof of the first part of the lemma: we now bound the change
in time-$(t+2)$ weights as we vary the covariances of time-$(t+1)$
social signals within this neighborhood.

Recall that we decomposed each covariance $\Covar(r_{i,t}-\theta_{t-1},r_{j,t}-\theta_{t-1})=\psi_{kk'}+\zeta_{ij}$
into a term $\psi_{kk'}$ depending only on the types of the two agents
and a term $\zeta_{ij}$, and similarly for variances. To show the
covariance matrix is contained in $\mathcal{W}_{\psi,\zeta'}$, we
bound each of these terms suitably.

We begin with $\zeta_{ij}$ (and $\zeta_{i}$). We can write 
\[
r_{i,t+1}=\sum_{j}\frac{W_{ij,t+1}}{1-w_{i,t+1}^{s}}\rho a_{i,t}=\sum_{j}\frac{W_{ij,t+1}}{1-w_{i,t+1}^{s}}\rho \left(w_{j,t}^{s}s_{j,t}+(1-w_{j,t}^{s})r_{j,t}\right).
\]
By the first part of the lemma, the ratio between any two weights
(both of the form $W_{ij,t+1}$, $w_{i,t+1}^{s},$ or $w_{j,t}^{s}$)
corresponding to pairs of agents of the same types is in $[1-\gamma_{1}/n,1+\gamma_{1}/n]$
for a constant $\gamma_{1}$. We can use this to bound the variation
in covariances of $r_{i,t+1}$ within types by $\zeta'$: we take
the covariance of $r_{i,t+1}$ and $r_{j,t+1}$ using the expansion
above and then bound the resulting summation by bounding all coefficients.

Next we bound $\psi_{kk'}$ (and $\psi_{k}$). It is sufficient to
show that $\text{Var}(r_{i,t+1}-\theta_{t})$ is at most $\psi$.
To do so, we will give an estimator of $\theta_{t}$ with variance
less than $\beta/n$, and this will imply $\text{Var}(r_{i,t+1}-\theta_{t})<\beta/n=\psi$
(recall $r_{i,t+1}$ is the estimate of $\theta_{t}$ given agent
$i$'s social observations in period $t+1$). Since this bounds all
the variance terms by $\psi$, the covariance terms will also be bounded
by $\psi$ in absolute value.

Fix an agent $i$ of network type $k$ and consider some network type
$k'$ such that $p_{kk'}>0$. Then there exists two signal types,
which we call $A$ and $B$, such that $i$ observes $\Omega(n)$
agents of each of these signal types in $G_{n}^{k}$.\footnote{ We use the notation $\Omega(n)$ to mean greater than $Cn$ for some
constant $C>0$ when $n$ is large.} The basic idea will be that we can approximate $\theta_{t}$ well
by taking a linear combination of the average of observed agents of
network type $k$ and signal type A and the average of observed agents
of network type $k$ and signal type B.

In more detail: Let $N_{i,A}$ be the set of agents of type A in network
type $k$ observed by $i$ and $N_{i,B}$ be the set of agents of
type B in network type $k$ observed by $i$. Then fixing some agent
$j_{0}$ of network type $k,$
\[
\frac{1}{|N_{i,A}|}\sum_{j\in N_{i,A}}a_{j,t-1}=\frac{\sigma_{A}^{-2}}{1+\sigma_{A}^{-2}}\theta_{t}+\frac{1}{1+\sigma_{A}^{-2}}r_{j_{0},t-1}+ \text{noise}
\]
where the noise term has variance of order $1/n$ and depends on signal
noise, variation in $r_{j,t}$, and variation in weights. These bounds
on the noise term follow from the assumption that the covariance matrix
of the $r_{i,t}$ is in a neighborhood $\mathcal{W}_{\psi,\zeta}$
and our analysis of variation in weights. Similarly
\[
\frac{1}{|N_{i,B}|}\sum_{j\in N_{i,B}}a_{j,t-1}=\frac{\sigma_{B}^{-2}}{1+\sigma_{B}^{-2}}\theta_{t}+\frac{1}{1+\sigma_{B}^{-2}}r_{j_{0},t-1}+ \text{noise}
\]
where the noise term has the same properties. Because $\sigma_{A}^{2}\neq\sigma_{B}^{2},$
we can write $\theta_{t}$ as a linear combination of these two averages
with coefficients independent of $n$ up to a noise term of order
$1/n$. We can choose $\beta$ large enough such that this noise term
has variance most $\beta/n$ for all $n$ sufficiently large. This
completes the proof of Lemma \ref{lem:2}.

\subsubsection{Step 4: }

We now give the two-step version of Lemma \ref{lem:WeightVariation1}.
\begin{lem}
\label{lem:WeightVariation2}Suppose that in period $t$ the matrix
$\bm{A}=\bm{A}_{t}$ of covariances of social signals satisfies $\bm{A}\in\mathcal{W}_{\psi,\zeta}$
and all agents are optimizing in periods $t+1$ and $t+2$. Then there
is a $\gamma$ so that for all $n$ sufficiently large, 
\[
\frac{W_{ij,t+2}}{W_{i'j',t+2}}\in\left[1-\frac{\gamma}{n},1+\frac{\gamma}{n}\right].
\]
whenever $i$ and $i'$ have the same network and signal types and
$j$ and $j'$ have the same network and signal types. 
\end{lem}
Given what we established about covariances in Step 3, the lemma follows
by the same argument as the proof of Lemma \ref{lem:WeightVariation1}.

\bigskip{}

\textbf{Step 5: }Now that Lemma \ref{lem:WeightVariation2} is proved,
we can apply it to show that 
$ \widetilde{\Phi}^{2}(\mathcal{W}_{\psi,\zeta})\subset\mathcal{W}_{\psi,\zeta}.$

We will do this by first writing the time-$(t+2)$ behavior in terms
of agents' time-$t$ observations (Step 5(a)), which comes from applying
$\widetilde{\Phi}$ twice. This gives a formula that can be used for
bounding the covariances\footnote{ We take this term to refer to variances, as well.}
of time-$(t+2)$ actions in terms of covariances of time-$t$ actions.
Step 5(b) then applies this formula to show we can take $\zeta_{ij}$
and $\zeta_{i}$ to be sufficiently small. (Recall the notation introduced
in Section \ref{subsec:Notation-and-key} above.) We split our expression
for $r_{i,t+2}$ into several groups of terms and show that the contribution
of each group of terms depends only on agents' types up to a small
noise term. Step 5(c) notes that we can also take $\psi_{kk'}$ and
$\psi_{k}$ to be sufficiently small.

\textbf{Step 5(a):} We calculate: 
\begin{eqnarray*}
r_{i,t+2} & = & \sum_{j}\frac{W_{ij,t+2}}{1-w_{i,t+2}^{s}}\rho a_{j,t+1}\\
 & = & \rho\left(\sum_{j}\frac{W_{ij,t+2}}{1-w_{i,t+2}^{s}}w_{j,t+1}^{s}s_{j,t+1}+\sum_{j,j'}\frac{W_{ij,t+2}}{1-w_{i,t+2}^{s}}W_{jj',t+1}\rho a_{j',t}\right)\\
 & = & \rho \Bigg( \sum_{j}\frac{W_{ij,t+2}}{1-w_{i,t+2}^{s}}w_{j,t+1}^{s}s_{j,t+1}+\rho\Big(\sum_{j,j'}\frac{W_{ij,t+2}}{1-w_{i,t+2}^{s}}W_{jj',t+1}w_{j',t}^{s}s_{j',t}\\
 &  & +\sum_{j,j'}\frac{W_{ij,t+2}}{1-w_{i,t+2}^{s}}W_{jj',t+1}(1-w_{j',t}^{s})r_{j',t}\Big) \Bigg).
\end{eqnarray*}
Let $h_{ij',t}$ be the coefficient on $r_{j',t}$ in this expansion
of $r_{i,t+2}$. Explicitly, 
\begin{equation}
h_{ij',t}=\sum_{j}\frac{W_{ij,t+2}}{1-w_{i,t+2}^{s}}W_{jj',t+1}(1-w_{j',t}^{s}).\label{eq:cijprimet}
\end{equation}
The coefficient $h_{ij',t}$ adds up the influence of $r_{j',t}$
on $r_{i,t+2}$ over all paths of length two.

First, we establish a lemma about how much these weights vary.
\begin{lem}
\label{lem:SSCoeffVariation}There exists $\gamma$ such that for
$n$ sufficiently large, when $i$ and $i'$ have the same network
types and $j'$ and $j''$ have the same network and signal types,
the ratio $h_{ij',t}/h_{i'j'',t}$ is in $[1-\gamma/n,1+\gamma/n]$.
\end{lem}
\begin{proof}
Fix $i$ and $j'$. For each network type \textbf{$k''$} and signal
type $\tau$, consider the number of agents $j$ of network type $k''$
and signal type $\tau$ who are observed by $i$ and who observe $j'$.
This number varies by at most a factor $\zeta^{2}$ as we vary $i$
and $j'$, preserving signal and network types. For each such $j$,
the contribution to $h_{ij',t}$ due to weight on that agent's action is (recalling
(\ref{eq:cijprimet}))
\[
\frac{W_{ij,t+2}}{1-w_{i,t+2}^{s}}W_{jj',t+1}(1-w_{j',t}^{s}).
\]

By applying Lemma \ref{lem:WeightVariation1} repeatedly, we can choose
$\gamma_{1}$ such that each of these contributions varies by at most
a factor of $\gamma_{1}/n$ as we change $i$ in $G_{k}$ and $j'$
in $G_{k'}$. Thus, $h_{ij',t}$ is a sum of terms which vary by at
most a multiplicative factor of $\gamma_{1}/n$ as we vary $i$
and $j'$ preserving signal and network types. If we can show that
the sum of the absolute values of these terms is bounded, then it
will follow that $h_{ij',t}$ varies by at most a multiplicative factor
of $\gamma/n$ for some $n$. This bound on the sum of absolute values
follows from the calculation of weights in the proof of Lemma \ref{lem:WeightVariation1}.
\end{proof}
\textbf{Step 5(b): }We first show that fixing the values of $\psi_{kk'}$
and $\psi_{k}$ in period $t$, the variation in the covariances $\Covar(r_{i,t+2}-\theta_{t+1,}r_{i',t+2}-\theta_{t+1})$
of these terms as we vary $i$ and $i'$ over network types is not
larger than $\zeta$. From the formula above, we observe that we can
decompose $r_{i,t+2}-\theta_{t+1}$ as a linear combination of three
mutually independent groups of terms:

(i) signal error terms $\eta_{j,t+1}$ and $\eta_{j',t}$;

(ii) the errors $r_{j',t}-\theta_{t}$ in the social signals from
period $t$; and

(iii) changes in state $\nu_{t}$and $\nu_{t+1}$ between periods
$t$ and $t+2$.

Note that the terms $r_{j',t}-\theta_{t}$ are linear combinations
of older signal errors and changes in the state. We bound each of
the three groups in turn:

\textbf{(i) Signal errors: }We first consider the contribution of
signal errors. When $i$ and $i'$ are distinct, the number of such
terms is close to its expected value because we are conditioning on
the events $E_{2}$ and $E_{4}$ defined in Section \ref{subsec:Notation-and-key}.
Moreover the weights are close to their expected values by Step 2,
so the variation is bounded suitably. When $i$ and $i'$ are equal,
we use the facts that the weights are close to their expected values
and the variance of an average of $\Omega(n)$ signals is small.

\textbf{(ii) Social signals: }We now consider terms $r_{j',t}-\theta_{t}$,
which correspond to the third summand in our expression for $r_{i,t+2}$.
Since we will analyze the weight on $\nu_{t}$ below, it is sufficient
to study the terms $r_{j',t}-\theta_{t-1}.$

By Lemma \ref{lem:SSCoeffVariation}, the coefficients placed on $r_{j',t}$
by $i$ and on $r_{j'',t}$ by $i'$ vary by a factor of at most $2\gamma/n$.
Moreover, the absolute value of each of these covariances is bounded
above by $\psi$ and the variation in these terms is bounded above
by $\zeta$. We conclude that the variation from these terms has order
$1/n^{2}$.

\textbf{(iii) Innovations: }Finally, we consider the contribution
of the innovations $\nu_{t}$ and $\nu_{t+1}$. We treat $\nu_{t+1}$
first. We must show that any two agents of the same types place the
same weight on the innovation $\nu_{t+1}$ (up to an error of order
$\frac{1}{n^{2}}$). This will imply that the contributions of timing
to the covariances $\Covar(r_{i,t+2}-\theta_{t+1,}r_{i',t+2}-\theta_{t+1})$
can be expressed as a term that can be included in the relevant $\psi_{kk'}$
and a lower-order term which can be included in $\zeta_{ii'}$.

The weight an agent places on $\nu_{t+1}$ is equal to the weight
she places on signals from period $t+1$. So this is equivalent to
showing that the total weight 
\[
\rho\sum_{j}\frac{W_{ij,t+2}}{1-w_{i,t+2}^{s}}w_{j,t+1}^{s}
\]
agent $i$ places on period $t+1$ depends only on the network type
$k$ of agent $i$ and $O_p(1/n^{2})$ terms. We will first show the
average weight placed on time-$(t+1)$ signals by agents of each signal
type depends only on $k$. We will then show that the total weights
on agents of each signal type do not depend on $n$.

Suppose for simplicity here that there are two signal types $A$ and
$B$; the general case is the same. We can split the sum from the
previous paragraph into the subgroups of agents with signal types
$A$ and $B$: 
\[
\rho\sum_{j:\sigma_{j}^{2}=\sigma_{A}^{2}}\frac{W_{ij,t+2}}{1-w_{i,t+2}^{s}}w_{j,t+1}^{s}+\rho\sum_{j:\sigma_{j}^{2}=\sigma_{B}^{2}}\frac{W_{ij,t+2}}{1-w_{i,t+2}^{s}}w_{j,t+1}^{s}.
\]
Letting $W_{i}^{A}=\sum_{\sigma_{j}^{2}=\sigma_{A}^{2}}\frac{W_{ij,t+2}}{1-w_{i,t+2}^{s}}$
be the total weight placed on agents with signal type $A$ and similarly
for signal type $B$, we can rewrite this as:

\[
W_{i}^{A}\rho\sum_{j:\sigma_{j}^{2}=\sigma_{A}^{2}}\frac{W_{ij,t+2}}{W_{i}^{A}(1-w_{i,t+2}^{s})}w_{j,t+1}^{s}+W_{i}^{B}\rho\sum_{j:\sigma_{j}^{2}=\sigma_{B}^{2}}\frac{W_{ij,t+2}}{W_{i}^{B}(1-w_{i,t+2}^{s})}w_{j,t+1}^{s}.
\]
The coefficients $\frac{W_{ij,t+2}}{W_{i}^{A}(1-w_{i,t+2}^{s})}$
in the first sum now sum to one, and similarly for the second. We
want to check that the first sum $\sum_{j:\sigma_{j}^{2}=\sigma_{A}^{2}}\frac{W_{ij,t+2}}{W_{i}^{A}(1-w_{i,t+2}^{s})}w_{j,t+1}^{s}$
does not depend on $k$, and the second sum is similar.

For each $j$ in group $A$, 
\[
w_{j,t+1}^{s}=\frac{\sigma_{A}^{-2}}{\sigma_{A}^{-2}+(\rho^{2}\kappa_{j,t+1}+1)^{-1}},
\]
where we define $\kappa_{j,t+1}^{2}=\text{Var}(r_{j,t+1}-\theta_{t})$
to be the error variance of the social signal. Because $\kappa_{j,t+1}$
is close to zero, we can approximate $w_{j,t+1}^{s}$ locally as a
linear function $\mu_{1}\kappa_{j,t+1}+\mu_{2}$ where $\mu_{1}<1$
(up to order $\frac{1}{n^{2}}$ terms).

So we can write the sum of interest as 
\[
\sum_{j:\sigma_{j}^{2}=\sigma_{A}^{2}}\frac{W_{ij,t+2}}{W_{i}^{A}(1-w_{i,t+2}^{s})}\left(\mu_{1}\sum_{j',j''}W_{jj',t+1}W_{jj'',t+1}(\rho^{2}\bm{V}_{j'j'',t}+1)+\mu_{2}\right).
\]
By Lemma \ref{lem:WeightVariation1}, the weights vary by at most
a multiplicative factor contained in $[1-\gamma/n,1+\gamma/n]$. The
number of paths from $i$ to $j'$ passing through agents of any network
type $k''$ and any signal type is close to its expected value (which
depends only on $i$'s network type), and the weight on each path
depends only on the types involved up to a factor in $[1-\gamma/n,1+\gamma/n]$.
The variation in $\bm{V}_{j'j'',t}$ consists of terms of the form
$\psi_{k'k''}$, $\psi_{k'}$, and $\zeta_{j'j''}$, all of which
are $O_p(1/n)$, and terms from signal errors $\eta_{j',t}$. The signal
errors only contribute when $j=j'$, and so only contribute to a fraction
of the summands of order $1/n$. So we can conclude the total variation
in this sum as we change $i$ within the network type $k$ has order
$1/n^{2}.$

Now that we know each the average weight on private signals of the
observed agents of each signal type depends only on $k$, it remains
to check that $W_{i}^{A}$ and $W_{i}^{B}$ only depend on $k$. The
coefficients $W_{i}^{A}$ and $W_{i}^{B}$ are the optimal weights
on the group averages 
\[
\sum_{j:\sigma_{j}^{2}=\sigma_{A}^{2}}\frac{W_{ij,t+2}}{W_{i}^{A}(1-w_{i,t+2}^{s})}\rho a_{j,t+1} \; \text{ and }  \; \sum_{j:\sigma_{j}^{2}=\sigma_{B}^{2}}\frac{W_{ij,t+2}}{W_{i}^{B}(1-w_{i,t+2}^{s})}\rho a_{j,t+1},
\]
so we need to show that the variances and covariance of these two
terms depend only on $k$. We check the variance of the first sum:
we can expand
\[
\sum_{\sigma_{j}^{2}=\sigma_{A}^{2}}\frac{W_{ij,t+2}}{W_{i}^{A}(1-w_{i,t+2}^{s})}\rho a_{j,t+1}=\sum_{\sigma_{j}^{2}=\sigma_{A}^{2}}\frac{W_{ij,t+2}}{W_{i}^{A}(1-w_{i,t+2}^{s})}\rho\left(w_{j,t+1}^{s}s_{j,t+1}+(1-w_{j,t+1}^{s})r_{j,t+1}\right).
\]
We can again bound the signal errors and social signals as in the
previous parts of this proof, and show that the variance of this term
depends only on $k$ up to error terms that are $O_p(1/n^{2})$. The second variance
and covariance are similar, so $W_{i}^{A}$ and $W_{i}^{B}$ depend
only on $k$ up to error terms that are $O_p(1/n^{2})$.

This takes care of the innovation $\nu_{t+1}$. Because we have included
any innovations prior to $\nu_{t}$ in the social signals $r_{j',t}$,
to complete Step 5(b) we need only show the weight on $\nu_{t}$ depends
only on the network type $k$ of an agent.

The analysis is a simpler version of the analysis of the weight on
$\nu_{t+1}$. It is sufficient to show the total weight placed on
period $t$ social signals depends only on the network type of $k$
of an agent $i$. This weight is equal to
\[
\rho^{2}\sum_{j,j'}\frac{W_{ij,t+2}}{1-w_{i,t+2}^{s}}\cdot W_{jj',t+1}\cdot(1-w_{j',t}^{s}).
\]
As in the $\nu_{t+1}$ case, we can approximate $(1-w_{j',t}^{s})$
as a linear function of $\kappa_{j',t}$ up to $O_p(1/n^{2})$ terms.
Because the number of paths to each agent $j'$ though a given type
and the weights on each such path cannot vary too much within types,
the same argument shows that this sum depends only on $k$, up to error terms that are $O_p(1/n^{2})$. Thus Step 5(b) is complete.

\textbf{Step 5(c): }The final step is to verify that we can take $\psi_{kk'}$
and $\psi_{k}$ to be smaller than $\psi$. It is sufficient to show
that the variance $\text{Var}(r_{i,t+2}-\theta_{t+1})$ of each social
signal about $\theta_{t+1}$ is at most $\psi$. The proof is the
same as in Step 2(b).

\includepdf[pages=-]{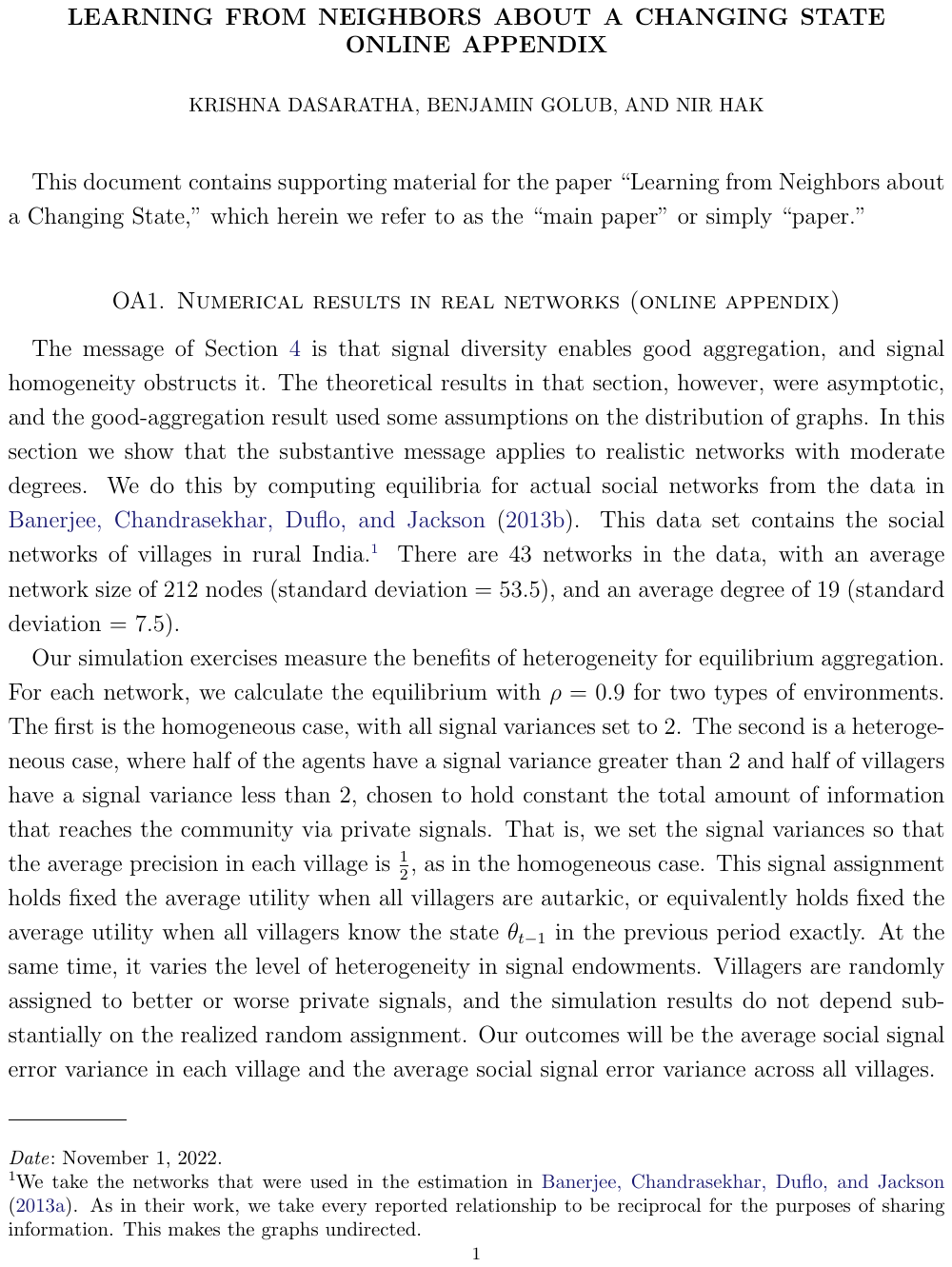}

\end{document}

%% file: supportingfiles/timeline.tex
\begin{tikzpicture}
\draw[-stealth,very thick, color={rgb:red,70;green,50;blue,197}] (0, 0) -- (\textwidth, 0);
\foreach \k in {0, 1, 2}
{
\draw[color={rgb:red,70;green,50;blue,197}	,ultra thick] (\textwidth/6+\k*\textwidth/3, -\textwidth/100) -- (\textwidth/6+\k*\textwidth/3, \textwidth/100);
\node at (\textwidth/6+\k*\textwidth/3, \textwidth/20) {\small $t=\the\numexpr\k-1\relax$};
\ifnum \k=0
\node at (\textwidth/6+\k*\textwidth/3-\textwidth/10, \textwidth/20){$\ldots$};
\draw[color={rgb:red,70;green,50;blue,197},very thick] (\textwidth/6+\k*\textwidth/3-\textwidth/8, -\textwidth/18) -- (\textwidth/6+\k*\textwidth/3, -\textwidth/18);
\node[left] at (\textwidth/6+\k*\textwidth/3-\textwidth/8, -\textwidth/18) {$\cdots$};
\fill[color=lightgray	,draw={rgb:red,70;green,50;blue,197}, very thick] (\textwidth/6+\k*\textwidth/3, -\textwidth/18) circle (3.5pt) node[color=black, below, align=center, text width=4cm]{\fontsize{7pt}{7.5pt}\selectfont$(i,-1)$ observes $t=-1$ signal\linebreak and acts\par};
\fi
\ifnum \k=1
\draw[color={rgb:red,70;green,50;blue,197}, very thick] ({\textwidth/6+(\k-1)*\textwidth/3}, -\textwidth/8) -- (\textwidth/6+\k*\textwidth/3, -\textwidth/8);
\fill[color=lightgray	,draw={rgb:red,70;green,50;blue,197}, very thick] ({\textwidth/6+(\k-1)*\textwidth/3}, -\textwidth/8) circle (3.5pt) node[color=black, below, align=center, text width=4cm]{\fontsize{7pt}{7.5pt}\selectfont$(i,0)$ born and observes\linebreak $t=-1$ actions\par};
\fill[color=lightgray	,draw={rgb:red,70;green,50;blue,197}, very thick] ({\textwidth/6+\k*\textwidth/3}, -\textwidth/8) circle (3.5pt) node[color=black, below, align=center, text width=4cm]{\fontsize{7pt}{7.5pt}\selectfont$(i,0)$ observes $t=0$ signal \linebreak and acts\par};
\fi
\ifnum \k=2
\node at (\textwidth/6+\k*\textwidth/3+\textwidth/10, 0.83){$\cdots$};
\draw[color={rgb:red,70;green,50;blue,197},very thick] ({\textwidth/6+(\k-1)*\textwidth/3}, -\textwidth/5) -- (\textwidth/6+\k*\textwidth/3, -\textwidth/5);
\fill[color=lightgray,draw={rgb:red,70;green,50;blue,197}, very thick] ({\textwidth/6+(\k-1)*\textwidth/3}, -\textwidth/5) circle (3.5pt) node[color=black, below, align=center, text width=4cm]{\fontsize{7pt}{7.5pt}\selectfont$(i,1)$ born and observes\linebreak $t=0$ actions\par};
\fill[color=lightgray,draw={rgb:red,70;green,50;blue,197}, very thick] ({\textwidth/6+\k*\textwidth/3}, -\textwidth/5) circle (3.5pt) node[color=black, below, align=center, text width=4cm]{\fontsize{7pt}{7.5pt}\selectfont$(i,1)$ observes $t=1$ signal\linebreak and acts\par};
\fi
}
\end{tikzpicture}